\documentclass[conference]{IEEEtran}
\IEEEoverridecommandlockouts
\usepackage{cite}
\usepackage{amsmath,amssymb,amsfonts}
\usepackage{algorithmic}
\usepackage{graphicx}
\usepackage{textcomp}
\usepackage{xcolor}
\def\BibTeX{{\rm B\kern-.05em{\sc i\kern-.025em b}\kern-.08em
    T\kern-.1667em\lower.7ex\hbox{E}\kern-.125emX}}

\usepackage{times,amsmath}
\usepackage{textcomp}
\usepackage{algorithmic}
\usepackage{diagbox}

\usepackage[linesnumbered,ruled,vlined]{algorithm2e} 
\usepackage{algorithm2e,setspace}

\usepackage{graphicx}
\usepackage{balance}  
\usepackage{amsmath}

\usepackage{subfigure}
\usepackage{graphicx}
\usepackage{booktabs}
\usepackage{amsmath,epsfig,latexsym,graphics,subfigure,latexsym}
\usepackage[normalem]{ulem}
\usepackage{url}
\usepackage{epsfig}

\usepackage{adjustbox}
\usepackage{framed}
\usepackage{pstricks}
\usepackage{xspace}
\usepackage{float}
\usepackage{latexsym}
\usepackage{color}
\usepackage{colortbl}
\usepackage{epstopdf}
\usepackage{caption}
\usepackage{multirow}
\usepackage{subfigure}
\usepackage{dsfont}
\usepackage{hhline}
\usepackage{balance}
\usepackage[euler]{textgreek}
\usepackage{caption}
\usepackage{mathrsfs}
\usepackage{cancel}
\usepackage{float}
\usepackage{amsthm}

\theoremstyle{plain}

\newtheorem{example}{{Example}}

\newtheorem{theorem}{{Theorem}}
\newtheorem{lemma}{{Lemma}}

\newtheorem{definition}{{Definition}}

\renewcommand{\baselinestretch}{0.98}

\newcommand{\rev}[1]{\textcolor[rgb]{0,0,0}{#1}}

\newcommand{\rrev}[1]{\textcolor[rgb]{0,0,0}{#1}}

\newcommand{\eat}[1]{}

\newcommand{\pull}{Pull-Based Paradigm}
\newcommand{\push}{Push-Based Paradigm}

\newcommand{\spc}{SPC\xspace}

\newcommand{\our}{PSPC}

\newcommand{\base}{HP-SP${\rm C_s}$}

\newcommand{\ourp}{PSC${\rm P^+}$}

\definecolor{mygray}{gray}{.9}
\definecolor{dkgreen}{rgb}{0,0.6,0}
\definecolor{gray}{rgb}{0.5,0.5,0.5}
\definecolor{mauve}{rgb}{0.58,0,0.82}

\begin{document}

\title{\our: Efficient Parallel Shortest Path Counting on Large-Scale Graphs}


\author{\IEEEauthorblockN{You Peng}
\IEEEauthorblockA{
\textit{The Chinese University of Hong Kong}\\
Hong Kong, China \\
ypeng@se.cuhk.edu.hk}
\and
\IEEEauthorblockN{Jeffrey Xu Yu}
\IEEEauthorblockA{
\textit{The Chinese University of Hong Kong}\\
Hong Kong, China \\
yu@se.cuhk.edu.hk}
\and
\IEEEauthorblockN{Sibo Wang}
\IEEEauthorblockA{
\textit{The Chinese University of Hong Kong}\\
Hong Kong, China \\
swang@se.cuhk.edu.hk}
}

\maketitle

\begin{abstract}
In modern graph analytics, the shortest path is a fundamental concept. Numerous \rrev{recent works} concentrate mostly on the distance of these shortest paths. Nevertheless, in the era of betweenness analysis, the counting of the shortest path between $s$ and $t$ is equally crucial. \rrev{It} is \rev{also} an important issue in the area of graph databases. In recent years, several studies have been conducted in an effort to tackle such issues. Nonetheless, the present technique faces a considerable barrier to parallel due to the dependencies in the index construction stage, hence limiting its application possibilities and wasting the potential hardware performance. To address this problem, we provide a parallel shortest path counting method that could avoid these dependencies and obtain approximately linear index time speedup as the number of threads increases. Our empirical evaluations verify the efficiency and effectiveness.

\end{abstract}


\section{Introduction}
\label{sect:intro}
The shortest path problem~\cite{sibowang1,sibowang2,sibowang3,peng1,peng10,peng11,peng12,peng13,peng14} is a basic one in the field of graph analytics~\cite{10.1145/3448016.3457237,shi2022indexing,peng15}, and numerous studies have been conducted on its related topics, e.g., shortest path distance~\cite{abraham2011hub,akiba2012shortest,akiba2014dynamic,li2019scaling}, shortest path counting~\cite{zhang2020hub,oyama2004applying,ren2018shortest,botea2021counting}. Given two vertices $s$ and $t$, a shortest path between them is the one with the shortest length among all possible paths. A variety of applications, e.g., geographic navigation, Internet routing, socially tenuous group detecting~\cite{shen2017finding}, influential community searching~\cite{li2017most}, event detection~\cite{rozenshtein2014event}, betweenness centrality~\cite{brandes2001faster,puzis2007fast}, and route planning~\cite{abraham2011hub,abraham2012hierarchical}, make extensive use of it. In the aforementioned applications, distance between $s$ and $t$ is frequently taken into account when determining the vertices' importance and relevance, e.g., the nearest keyword search~\cite{jiang2015exact}, and ranking search in social networks~\cite{vieira2007efficient,peng2,peng3,peng4,peng5,peng6,peng7,peng8,peng9,peng10}.

Beside distance, the shortest path counting (\spc) is a vital aspect of the shortest path related problems. This is due to the fact that basing relevance and importance merely on distance is uninformative~\cite{zhang2020hub}. In addition, many graphs in the real world have a limited diameter owing to the small-world phenomenon. Consequently, numerous pairs of vertices have the same distance between them. Several pairings of vertices will be considered equally relevant based on the distance information alone. This is really unrealistic. A typical example is given as follows: Consider graph $H$ in Figure~\ref{fig:motivation}. Both $t_1$ and $t_2$ are located at a distance of $2$ from $s$ in $H$. Therefore, $t_1$ and $t_2$ are considered equally important to $s$ based only on their distance. However, such a conclusion is irrational, since $s$ and $t_2$ are connected by the shortest paths and hence are more relevant. In light of this, it is also essential to count the number of shortest paths between any two specified vertices. 

Even with a not-so-``small" diameter (the longest shortest path), the distance is also uninformative. Consider a connected unweighted graph $G$ with $n$ vertices and a diameter of $d$. Given a vertex $v$, it could only differentiate $d$ types of vertices if basing merely on distance. \rrev{In the most of real networks,} e.g., SNAP~\cite{snapnets}, $d \ll n$.



\begin{figure}[htb]
	\centering
	\includegraphics[width=0.40\linewidth]{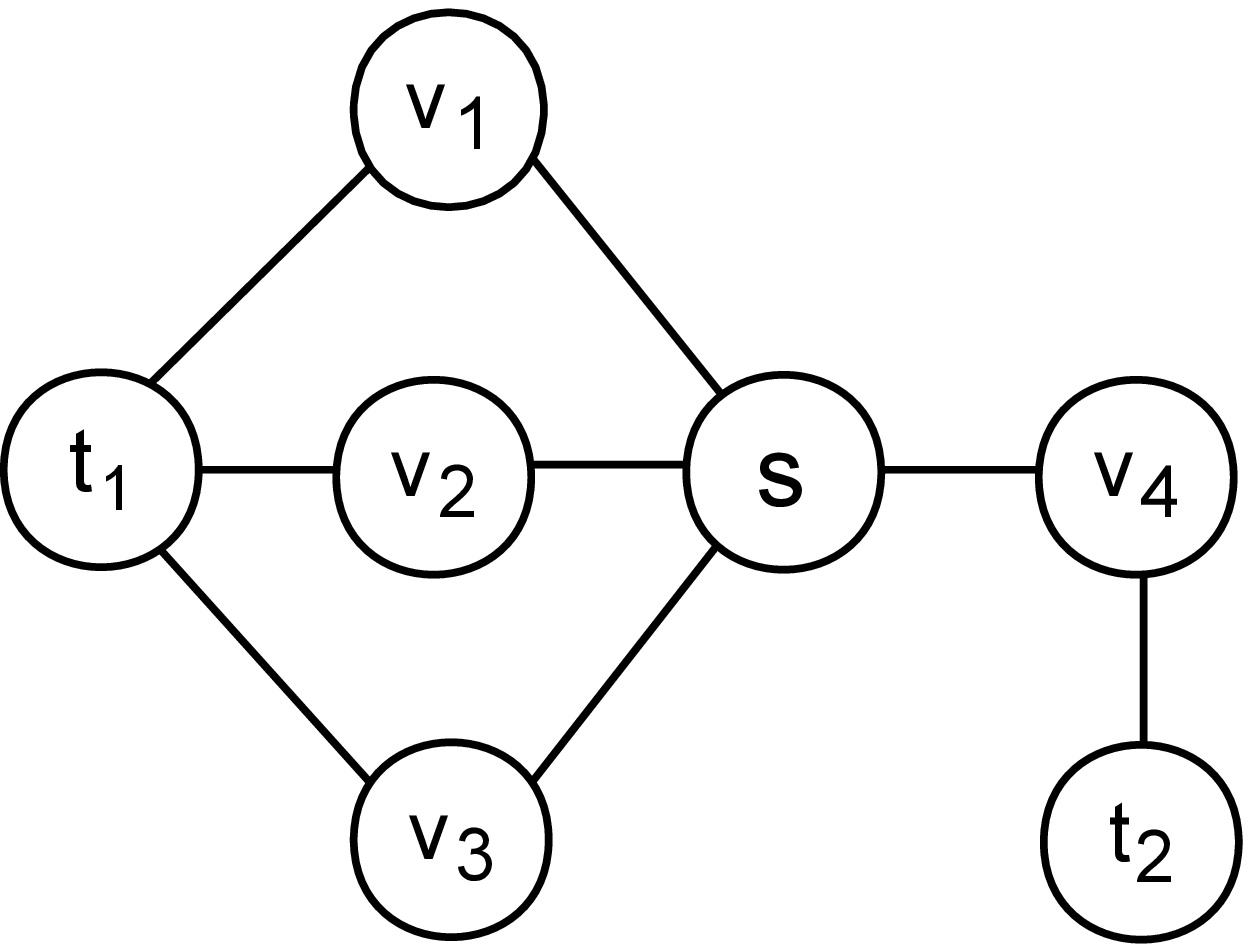}
	\caption{Graph H}
	\label{fig:motivation}
\end{figure}

\noindent \textbf{Application.} Listed below are some vital applications for the \spc problem.

\noindent(1) \underline{\textit{Group Betweenness}}.~The group betweenness exemplifies a typical application of the \spc problem. Puzis et al.~\cite{puzis2007fast} estimate the importance of a vertex set $C$ to $G$ based on the group betweenness. Let $P_{s, t}$ represent the set of shortest paths connecting vertices $s$ and $t$, $spc(s, t)$ represent the number of shortest paths connecting vertices $s$ and $t$, and $spc_C(s, t)$ indicate the $\#$ of the paths in $P_{s, t}$ via $C$. Group Betweenness of $C$, indicated by \"B$(C)$, is defined as \"B$(C) = \sum_{s, t}spc_{C}(s, t) / spc(s, t)$. As shown in~\cite{puzis2007fast}, there is an algorithm GBC for progressively evaluating \"B$(C)$. In particular, let $C = \{v_1,...,v_{|C|} \}$ and let $C_i = \{v_1,...,v_i \}$. GBC calculates \"B$(C_i)$ in the $i$-th iteration by adding to \"B$(C_{i-1})$ the total fraction of shortest paths that pass through $v_i$ but not $C_{i-1}$. Therefore, after $|C|$ iterations, \"B$(C)$ is found. 

GBC receives as input three $|C| \times |C|$ matrices $D$, $\sum$, and $\widetilde{B}$, which store for $\forall x, y \in C$ the distance between $x$ and $y$, $spc(x,y)$, and the path betweenness of $(x, y)$, respectively. After computing \"B$(C_i)$, GBC modifies $\widetilde{B}$ depending on $D$ and $\sum$ such that $\widetilde{B}_{v_{i+1},v_{i+1}} =$ \"B$(C_{i+1}) - $ \"B$(C_i)$. This makes it simple to compute \"B$(C_{i+1})$ in the subsequent iteration. Overall, GBC needs $O(T+|C|^3)$ time, where $T$ is the construction time for $D$, $\sum$, and $\widetilde{B}$. In tasks such as estimating group betweenness distribution, the $\#$ of groups to evaluate is vast. To reduce the online time necessary to generate $D$, $\sum$, and $\widetilde{B}$,~\cite{puzis2007fast} proposes to pre-compute and store the distance, the $\#$ of shortest paths, and the path betweenness for every pair of vertices, incurring prohibitive overhead. Although distance hub labeling and VC-dimension-based techniques~\cite{riondato2016fast} may minimize the cost associated with distance and path betweenness, the cost regarding \spc remains intractable.

\noindent(2) \underline{\textit{Road Networks}}.~\rrev{In real-world road network applications, the greater the number of shortest routes, the greater the number of traffic possibilities and the flexibility of route planning from the origin to the destination. For example, the top-$k$ nearest neighbours search attempts to locate $k$ objects adjacent to the query vertex inside a candidate set. It is a leading provider of taxi (e.g., Uber), restaurant (e.g., Tripadvisor), and hotel (e.g., Booking) recommendation services. A candidate item may be more desirable than others with the same or comparable distance if many shortest paths go to it, since we have more backup routing options and a greater chance of avoiding traffic congestion. In a movie ticket application, for instance, there are two theatres with the same shortest distance to the source location. Given the available traffic possibilities, we may choose the route with the most shortest paths. In addition to acting as a proximity measure, the shortest path count has been employed as a building component in the computation of betweenness centrality~\cite{pontecorvi2015faster,puzis2007fast}.
}

\noindent \rrev{\textbf{Challenges.}~In the above application, the main obstacle is the increasing graph size of real applications. At the current stage, large-scale graphs are common for real applications. Nevertheless, the state-of-the-art algorithm for building an index based on node orders limits its scalability. To address this issue, we carefully designed a different scalable algorithm for this problem. Our experiments demonstrate that our algorithm could build an index for large-scale graphs. In addition, the query time is scalable.}

\noindent \rrev{\textbf{Contributions.}~The primary contributions are listed below.}

\noindent \rrev{\textit{\underline{(1) Parallel Algorithm}}. This paper begins by analyzing the interdependence of the current hub labeling approach and explaining the challenge of parallel processing. It redesigns another propagation mechanism to avoid these dependencies by detecting the dependencies.}

\noindent \rrev{\textit{\underline{(2) Acceleration Optimizations}}.~This paper investigates scheduling planning and landmark-based labeling to boost efficiency. Using these optimization techniques, our index creation phase might be accelerated significantly.}

\noindent \rrev{\textit{\underline{(3) Hybrid Vertex Ordering}}.~Classical vertex orders for the SPC involve ordering by the degree and ordering by the significant path. In road networks, the significant-path-based ordering surpasses the degree-based ordering. Nevertheless, significant-path-based ordering must select the next vertex based on the shortest path trees built in the current vertex's construction, implying a dependency for each vertex and leading parallel processing challenging. This paper proposes a vertex-based ordering for the road network and a degree-based ordering for the social network to fill this gap. Then, it mixes them to provide a hybrid ordering for a scalable index construction process.}

\noindent \rrev{\textit{\underline{(4) Comprehensive Experimental Evaluation}}.~In order to illustrate the effectiveness and efficiency of our algorithms, $10$ datasets are employed in this study. The experimental results indicate that our method outperforms the baselines in terms of index building time and generates comparable index sizes. Furthermore, our method scales approximately linearly with the number of threads.}

\noindent \textbf{Roadmap.}~The rest of the paper is organized as follows. Section~\ref{sect:related} presents important related works. Section~\ref{sect:pre} introduces the preliminaries of the \spc problem. Section~\ref{sect:exp} experimentally evaluates our proposed approaches on real small-world networks and Section~\ref{sect:con} concludes the paper.
\section{Preliminaries}
\label{sect:pre}
\noindent \textbf{Graphs}.~This paper concentrates on an unweighted and undirected graph denoted by $G=(V, E)$, where $V$ and $E$ represent the set of vertices and edges in $G$, respectively. Let $n$ = $|V|$ and $m = |E|$ denote the number of vertices and the number of edges, respectively. For each vertex $v \in V$, let nbr($v$) be the set of $v's$ neighbors and deg($v$) be the degree of $v$. A path $p$ from vertex $s$ to vertex $t$ is defined as a sequence of vertices $(s=v_0,v_1,...,v_l=t)$ such that $(v_i, v_{i+1}) \in E$ for $0 \leq i < l$. The length of $p$, denoted by $len(p)$, is the number of edges included in $p$. In other words, $len(p) = l$. For simplicity, this work uses the notation $rev(p)$ to denote the reverse of a path. Specifically, $rev(p) = (v_l,v_{l-1},...,v_0)$. A path from $s$ to $t$ is shortest if its length is no larger than any other path from $s$ to $t$. 

The notations are summarized in Table~\ref{tb:notations}.

\begin{table}[tb]
\small
  \centering
  \vspace{-1mm}
  \caption{The summary of notations}
\label{tb:notations}
    \begin{tabular}{|c|l|}
      \hline
      \cellcolor{gray!25}\textbf{Notation} & \cellcolor{gray!25}\textbf{Definition}        \\ \hline
      $G = (V, E)$ &  a given undirected and unweighted graph \\ \hline
      $m$, $n$ & the number of edges(vertices) for $G$ \\ \hline
      $nbr(v)$ & the set of neighbors of $v$ \\ \hline
      $deg(v)$ & the degree of $v$ \\ \hline
      $len(p)$ & the length of a path $p$ \\
               & $len(p) = l$ if $p = (v_0, v_1,..., v_l)$ \\ \hline
      $rev(p)$ & the reverse of a path $p$ \\
               & $rev(p) = (v_l, v_{l-1},...,v_0)$ if $p=(v_0,v_1,...,v_l)$ \\ \hline
      $p_{s,t}$& a shortest path from $s$ to $t$ \\ \hline
      $P_{s,t}$& the set of shortest paths from $s$ to $t$ \\ \hline
      $SPC_{s,t}$& the shortest path counting from $s$ to $t$ \\ \hline
      $L^C, L^{NC}$ & the canonical and non-canonical labels \\ \hline
      $dis(u, w)$ & the distance from vertex $u$ to $w$ \\ \hline
      $L(v)$ & ESPC index for vertice $v$ \\ \hline 
      $C_{v,w}$ & the trough path counting from $v$ to $w$ \\ \hline
      $rev(p)$ & the reverse path of $p$ \\ \hline
     \end{tabular}

\vspace{-1mm}
\end{table}

\begin{figure*}[htb]
    \centering
    \begin{minipage}{.4\linewidth}
        \centering
        \includegraphics[scale=0.3]{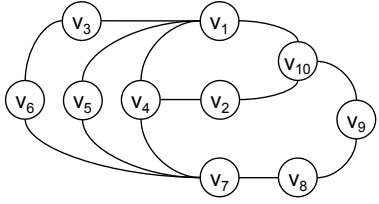}
        \caption{An Undirected Graph. The total order $\leq$ is $v_1 \leq v_7 \leq v_4 \leq v_{10} \leq v_3 \leq v_5 \leq v_6 \leq v_2 \leq v_8 \leq v_9$.}
        \label{fig:ori_g}
    \end{minipage}%
    \begin{minipage}{.7\linewidth}
        \centering
        \captionof{table}{Shortest Path Counting Labels of Fig.~\ref{fig:ori_g}}
        \vspace{-2mm}
        \scalebox{0.8}{
  \begin{tabular}{|l|l|l|}
    \hline
    \cellcolor{gray!25}\textbf{Vertex} & \cellcolor{gray!25}\textbf{$L(\cdot)$}\\
    \hline
    $v_{1}$ & $(v_{1},0,1)$\\
    \hline
    $v_{2}$ & $(v_{1},2,2)$ $(v_{7},2,1)$ $(v_{4},1,1)$ $(v_{10},1,1)$ $(v_{2},0,1)$\\
    \hline
    $v_{3}$ & $(v_{1},1,1)$ $(v_{7},2,1)$ $(v_{3},0,1)$\\
    \hline
    $v_{4}$ & $(v_{1},1,1)$ $(v_{7},1,1)$ $(v_{4},0,1)$\\
    \hline
    $v_{5}$ & $(v_{1},1,1)$ $(v_{7},1,1)$ $(v_{5},0,1)$\\
    \hline
    $v_{6}$ & $(v_{1},2,1)$ $(v_{7},1,1)$ $(v_{3},1,1)$ $(v_{6},0,1)$\\
    \hline
    $v_{7}$ & $(v_{1},2,2)$ $(v_{7},0,1)$\\
    \hline
    $v_{8}$ & $(v_{1},3,3)$ $(v_{7},1,1)$ $(v_{10},2,1)$ $(v_{8},0,1)$\\
    \hline
    $v_{9}$ & $(v_{1},2,1)$ $(v_{7},2,1)$ {$(v_{4},3,1)$} $(v_{10},1,1)$ $(v_{8},1,1)$ $(v_{9},0,1)$\\
    \hline
    $v_{10}$ & $(v_{1},1,1)$ {$(v_{7},3,2)$ $(v_{4},2,1)$} $(v_{10},0,1)$\\
    \hline
\end{tabular}
}
        \label{tab:labels}
    \end{minipage}
\vspace{-1mm}    
\end{figure*}

\subsection{2-Hop Labeling for Shortest Path Counting}
\label{subsect:2hopCounting}
To efficiently process point-to-point SPC queries, the 2-hop labeling technique~\cite{zhang2020hub} precomputes the SPC information from each node to pre-selected hub nodes and utilizes the $2$-hop via hubs to respond to a query. It presents the hub labeling method for the shortest path counting between vertices $s$ and $t$, SPC($s,t$). 
The shortest path counting between vertices $s$ and $t$ in a directed graph seeks to determine the total number of all the shortest paths from $s$ to $t$. \cite{zhang2020hub} proposed a 2-hop labeling scheme and an index construction algorithm to build the index efficiently and enable real-time shortest path counting queries. 
\rev{The hub labeling scheme supports the cover constraint, \underline{E}xact \underline{S}hortest \underline{P}ath \underline{C}overing (ESPC), which implies that it not only encodes the shortest distance between two vertices but also ensures that such shortest paths are correctly counted. HP-SPC is the algorithm developed for constructing the SPC label index that satisfies ESPC.}

Formally, given a directed graph $G$, \rev{HP-SPC} assigns each vertex $v \in G$ an in-label $L_{in}(v)$ and an out-label $L_{out}(v)$, consisting of entries of the form $(w, sd(v,w), \theta_{v,w})$. The shortest distance between $v$ and $w$ is denoted by $sd(v,w)$, and the number of shortest paths between $v$ and $w$ is denoted by $\theta_{v,w}$. If $w \in L_{in}(v)$ or $w \in L_{out}(v)$, then $w$ is considered a hub of $v$. 

In essence, the in-label $L_{in}(v)$ keeps track of the distance and counting information from its hubs to itself, whereas the out-label $L_{out}(v)$ records distance and counting information from $v$ to its hubs. \rev{HP-SPC} adheres to the cover constraint, which states that for each given starting vertex $s$ and ending vertex $t$, there exists a vertex $w \in L_{out}(s) \cap L_{in}(t)$ that lies on the shortest path from $s$ to $t$. 

In addition, \rev{HP-SPC} guarantees the correctness of shortest path counting by including the shortest path from $s$ to $t$ via a hub vertex once during the label construction. 
SPC($s,t$) is evaluated by scanning the $L_{out}(s)$ and $L_{in}(t)$ for the shortest distance via common hubs and adding the multiplication of the corresponding count. Equation \eqref{con:spc1} identifies all common hubs (on the shortest paths) from $L_{out}(s)$ and $L_{in}(t)$. Equation \eqref{con:spc2} determines the result of SPC($s,t$).

\begin{equation}
H = \{h | \mathop{\arg\min}\limits_{h \in L_{out}(s) \cap L_{in}(t)} \{sd(s,h) + sd(h,t)\}\} \label{con:spc1}
\end{equation}
\begin{equation}
SPC(s,t) = \sum_{h\in H} \theta(s,h) \cdot \theta(h,t) \label{con:spc2}
\end{equation}

\begin{example}
Figure~\ref{fig:ori_g} depicts \rrev{an undirected} graph with $10$ vertices, and Table~\ref{tab:labels} provides its hub labeling index for SPC queries. SPC($v_{10},v_{7}$) is used as an example to determine the shortest paths counting from $v_{10}$ to $v_{7}$. By scanning $L(v_{10})$ and $L(v_{7})$, two common hubs \{$v_{1}, v_{7}$\} are found. \rrev{The shortest distance through $v_{1}$ is $2 + 2 = 4$, whereas the counting is $2$ $\cdot$ $1 = 2$; The shortest distance via $v_{7}$ is $3 + 0 = 3$, and the counting is $1$ $\cdot$ $2 = 2$. Therefore, the number of shortest paths from $v_{10}$ to $v_{7}$ is $2 + 2 = 4$ with a length of $4$}. 
\end{example}

\section{\our \  algorithm description from the Covering Side}
\label{sect:tech1}


This section presents our parallel method for the \spc problem. First, the \textit{Shortest Path Covering} is defined.

\begin{definition}(Shortest Path Covering).~$T(v)$ represents a set of entries of the form \rrev{$(w, C_{v, w})$}, where $C_{v, w} \subset P_{v,w}$ represents a subset of the shortest paths from $v$ to $w$. For each pair of vertices $u$ and $v$, the shortest paths are covered by $T(u)$ and $T(v)$ as a multiset as follows:
\begin{equation}
\begin{aligned}
   cover(T(u), T(v)) &= \\
   \{ p_1 \odot rev(p_2) &| (w, C_{u, w}) \in T(u), (w, C_{v,w}) \in T(v), \\
   & p_1 \in C_{u,w}, p_2 \in C_{v,w}, \\
   & sd(u, w) + sd(v, w) = sd(u, v) \}
\end{aligned}
\end{equation}
\rrev{It is noted that $w$ is the common vertex in both $T(u)$ and $T(v)$.}
\end{definition}

This is followed by the definition of \textit{Exact Shortest Path Covering}.
\begin{definition}(Exact Shortest Path Covering).~$T(\cdot)$ is an exact shortest path covering (ESPC for short), which means that for any two vertices $u$ and $v$, the multiset cover $(T(u), T(v))$ must be identical to $P_{u, v}$, i.e., the set of shortest paths between $u$ and $v$. 
\end{definition}

\noindent \rrev{\underline{\textit{The Construction of an ESPC Index.}}}~The construction process of an ESPC is explained. Let $\leq$ be a total order over $V$. A trough path~\cite{jiang2014hop} is a path whose endpoint is ranked higher than all the other vertices. A trough shortest path is a path that is both trough and shortest. \rrev{For a shortest path $p \in SP_{u,v}$, there exists a vertex $w$ with the highest rank. Thus, $p$ could be divided into two trough shortest paths $p_{u,w}$ and $p_{w,v}$. By doing this, we could find the exact shortest path covering by using the trough shortest paths of $u$ and $v$.}

Consider, for instance, the graph $G'$ in Figure~\ref{fig:ori_g}, and a total order $\leq$ where $v_1 \leq v_7 \leq v_4 \leq v_{10} \leq v_3 \leq v_5 \leq v_6 \leq v_2 \leq v_8 \leq v_9$. The path $(v_3 ,v_1, v_{10})$ is not a trough path since $v_1$ has a higher rank than both endpoints $v_3$ and \rrev{$v_{10}$}. The path $(v_6, v_3, v_1)$ is the trough shortest path because one endpoint
(i.e. $v_1$) has the highest rank and the path is the shortest. Given a total order $\leq$ over the vertices, an ESPC can be constructed as follows. $T(v)$ is initially empty for each vertex $v$. Then, for any two (potentially identical) vertices $v$ and $w$ with $w \leq v$, an entry $(w, C_{v, w})$ is added to $T(v)$, where $C_{v, w}$ is the set of all trough shortest paths from $v$ to $w$. This $C_{v, w}$ is not empty. Note that for every such label entry, since $w \leq v$, $w$ has the highest rank in $p$ for each path $p \in C_{v, w}$. $T_{\leq}(\cdot)$ and $L_{\leq}(\cdot)$ denote the $T(\cdot)$ constructed this way and the corresponding $L(\cdot)$, respectively.
The ESPC result in Figure~\ref{fig:ori_g} is shown in Table~\ref{tab:labels}.

Then, this work investigates why the state-of-the-art algorithm~\cite{zhang2020hub} cannot be parallelized. In the state-of-the-art algorithm, the label entries are separated into two types: canonical labels ($L^c$) and non-canonical labels ($L^{nc}$).

\rrev{$p_{s,t}$ denotes a shortest path in $G$ from $s$ to $t$. The shortest distance between $s$ and $t$ in $G$, denoted by $sd_G(s,t)$ is defined as the length of the shortest path between $s$ and $t$ in $G$. The set of all shortest paths from $s$ to $t$ is denoted by $P_{s,t}$, and the set of the vertices involved in $P_{s,t}$ is denoted by $Q_{s,t}$. $spc_G(s,t)$ denotes the number of shortest paths from $s$ to $t$ in $G$. When the context is clear, $sd(s,t)$ and $spc(s,t)$ are used instead of $sd_G(s,t)$ and $spc_G(s,t)$ for simplicity. Let $\leq$ be a total order over $V$. For two distinct vertices $w$ and $v$, if $w \leq v$, then $w$ has a higher rank than $v$. The total order of vertices is an order that the ESPC index constructs for each vertex. It could be obtained by ranking all the vertices in any order, e.g., sorting the vertices by the degree order.}

\rrev{\cite{zhang2020hub} defined the Canonical Hubs as follows: given a total order $\leq$ over the vertices, is one that comprises the following hubs: For two vertices $v$ and $w$, $w \in L(v)$ if and only if $w$ is the highest-ranked vertex in $Q_{v,w}$.}

\begin{definition}(Canonical Hubs).~A canonical hub labeling, given a total order $\leq$ over the vertices, is one that comprises the following hubs: For two vertices $v$ and $w$, $w \in L(v)$ if and only if $w$ is the highest-ranked vertex in $Q_{v,w}$.
\end{definition}
Thus, non-canonical hubs are those hubs that do not comply with the definition of Canonical hubs and also in ESPC.

\subsection{Trough Path Property}
\label{subsect:label_property}
The labels of $L^{c}$ demonstrate an essential node-order characteristic.
\begin{theorem}
\label{theo:label_property}
\rrev{For each pair of vertices $\forall u, v \in V$, $v$ is a hub of $u$ in $L^c$, i.e., $(v, dis(v,u), c(v,u)) \in L^c(u)$, if and only if paths $SP(u,...,v)$ are the trough paths between $u$ and $v$, i.e., $v$ is the highest-ranked node along all the shortest paths from $u$ to $v$.}
\end{theorem}
\begin{proof}
This property is proved by contradiction. $SP(u,..,v)$ represents the set of all shortest paths from $u$ to $v$. Consider a node $u_h$ that is the highest-ranked node in $SP(u,...,v)$. Assume that there exists a node $u_c$ in $SP(u,...,v)$ such that $u_c$ does not have a hub of $u_h$ in the label set. $L^c_{< u_h}$ denotes the label index when finishing the pruned BFS for all vertices whose rank is higher than $u_h$. Consider the construction of ESPC, and the iteration when the pruned BFS sourced $u_h$ is performing. When there is no hub of $u_h$ for $u_c$ in this iteration, then either
\begin{itemize}
    \item Query$(u_c, u_h, L^c_{< u_h}) = (d_1, c_1)$ and $d_1 < dis(u_c, u_h)$, or
    \item $u_c$ is not explored in the Pruned BFS sourced from $u_h$, indicating that there is a vertex $u_c'$ on the shortest paths from $u_c$ to $u_h$ and could be pruned with Query($u_c', u_h, L^c_{< u_h} = (d_2, c_2)$ ) with $d_2 < dis(u_h, u_c')$.
\end{itemize}

In either scenario, it needs a common hub between $u_h$ and $u_c$ to produce the query result. Otherwise, some shortest paths would be omitted from in our index, which would violate the definition of ESPC. Nevertheless, such a common hub cannot exist since i) $u_c, u_c' \in SP(u,...,v)$ and ii) $u_h$ is the highest-ranked vertex in $SP(u,...,v)$. This results in a contradiction.

Since all nodes in $SP(u,...,v)$ have $u_h$ as their hubs, the theorem can be proved in the two cases: i) $u_h = v$, i.e., $v$ is the highest-ranked vertex in $SP(u,...,v)$, the $v$ is a hub of $u$ and ii) if $r(u_h) > r(v)$, when before the pruned BFS sourced from $v$ is performed, $u_h$ is already a common hub of $u$ and $v$. Since $u_h$ is on the shortest path between $u$ and $v$, the label with hub $v$ on $u$ is inserted into the $cL$ or pruned.
\end{proof}

\subsection{Order Property}
\label{subsect:order_property}
The labels are partitioned in the index according to their hub nodes to see the dependency among the labels. Let $v_1 \leq v_2 \leq ... \leq v_n$ represent the node order under which label set the index was constructed.

Two distinct sets are defined. Recall that the \base \ index consists of $n$ iterations where the $i$-th iteration executes a pruned BFS sourced from $v_i$. $L^{SPC}_{<i}(u)$ is the snapshot of $L^{SPC}(u)$ at the beginning of the $i$-th iteration, and by $L^{SPC}_i(u)$ the incremental label of $u$ built in the $i$-th iteration. It is notable that $L^{SPC}_{<i}(u) = L^{c}_{<i}(u) \cup L^{nc}_{<i}(u)$

\begin{definition}(Order Specific Label Set).~$L^{SPC}(u) = (v_i, dis(v_i, u), c)$ $\in L^{SPC}$, for 
$\forall i \in [1, n]$, $u \in V$. Let $L^{SPC}_i = \bigcup_{u \in V} L^{SPC}_i(u)$.
\end{definition}

\begin{definition}(Order Partial Label Set).~$L^{SPC}_{<i}(u) = {(v_j, dis(v_j, u), c) \in L^{SPC} | j < i}$
for $\forall i \in [1, n+1]$, $u \in V$. Let $L^{SPC}_{<i} = \bigcup_{u \in V} L^{SPC}_{<i}$. $L^{SPC}_{<n+1} = L^{SPC}$.
\end{definition}

The following lemma demonstrates that the pruning condition in \base \ results in an order dependency among labels.

\begin{lemma}[Order Dependency]
\label{lemma:order_depend}
The $L^{SPC}_i$ would depend on $L^{SPC}_{<i}(u)$. Specifically, $L^{SPC}_i(u)$ would be updated as follows:
\begin{itemize}
    \item If $Query(v_i, u, L^{SPC}_{<i}) = (d_1, c_1)$, and $d_1 > dis(v_i, u)$, then $L^c_i(u) = {(v_i, dis(v_i, u), count(v_i, u))}$.
    \item If $Query(v_i, u, L^{SPC}_{<i}) = (d_1, c_1)$, and $d_1 = dis(v_i, u)$, then $L^{nc}_i(u) = {(v_i, dis(v_i, u), count(v_i, u) + c_1)}$.
    \item Otherwise, $L^{SPC}_i(u) = \emptyset$.
\end{itemize}

\end{lemma}
\begin{proof}
Let $S$ be the set of nodes on the shortest path from $v_i$ to $u$ (including $v_i$ and $u$). Let $w$ be the node with the highest rank in $S$. If $v_i = w$, according to Theorem~\ref{theo:label_property}, i) $v_i$ is a hub of $u$ in $L^c$ and ii) for $\forall v \in S \ v_i$, $v$ is not a hub of $v_i$ in $L^{SPC}$, and hence, if Query($v_i, u, L^{SPC}_{<i} = (d_1, c_1)$), then $d_1 > dis(v_i, u)$. If $r(v_i) < r(w)$, then $v_i$ is not a hub of $u$ in $L^c$ and label $(w, dis(w, v_i), d_{w1}), (w, dis(w, u), d_{w2}) \in L^{SPC}_{<i}$ and thus Query$(v_i, u, L^{SPC}_{<i}) = (d_2, c_2)$, then $d_2 = dis(v_i, u)$. It is only necessary to update $L^{nc}_{\leq i}$ in the $i$-th iteration.
\end{proof}

Lemma~\ref{lemma:order_depend} demonstrates that $L^{SPC}_i(u)$ depends on $L^{SPC}_{<i}$ while $L^{SPC}_{<i}(u)$ depends on $L^{SPC}_{i-1}$. Such a convolved dependency is difficult to remove so long as the labels are constructed in the node order.

\subsection{Distance Dependency}
\label{subsect:dis_depend}
\rrev{To break the order dependency in the label construction, we consider the pruning condition where Query$(v_i, u, L^{SPC}_{<i}) = (dis(u, v_i), count(u, v_i))$. It prunes a node label on $u$, and there must be two labels on $u$ and $v_i$ to a common hub $w$ such that $dis(u,w) + dis(w, v_i) < dis(u, v_i)$.} Consequently, $dis(u, w)$ and $dis(w, v_i)$ must be smaller than $dis(u, v_i)$. In other words, none of the labels with distances greater than $dis(u, v_i)$ influence the query result of Query$(v_i, u, L^{SPC}_{<i})$ and the corresponding pruning outcomes. 

Based on the above intuition, the label entries in $L^{SPC}$ are categorized by their label distances. \rrev{The reorganized label sets will pave the way to our \underline{P}arallel \underline{S}hortest \underline{P}ath \underline{C}ounting Labeling approach and are hence referred to as PSPC label sets.} Let $D$ be the diameter of graph $G$.

\begin{definition}
\label{def:DSLS}
(Distance Specific Label Set).~$L^{PSPC}_{d}(u)$ $=$ ${(u, dis(v,u), c)} \in L^{SPC}(v) | dis(v, u) = d$,
for $\forall u \in V, d \in [1,D]$. Let $L^{PSPC}_d = {L^{SPC}_d(u) | u \in V}$.
\end{definition}

Similarly, the partial label of a node then becomes the set of label entries with a distance no larger than a certain distance and is defined in Definition~\ref{def:DPLS}.
\begin{definition}
\label{def:DPLS}(Distance Partial Label Set).~$L^{PSPC}_{\leq d}(u)$ $=$ ${(v, dis(v, u), c) \in L^{SPC}(u) | dis(v, u) \leq d}$,
for $\forall u \in V$, $d \in [1, D+1]$. Let $L^{PSPC} = \bigcup_{u \in V}L^{PSPC}(u)$. In particular, $L^{PSPC}(u) = L^{PSPC}_{\leq D+1}(u)$.
\end{definition}

Theorem~\ref{theo:equal} describes the equivalence between the index $L^{SPC}$ and the new index $L^{PSPC}$.
\begin{theorem}
\label{theo:equal}
$L^{SPC} = L^{PSPC}$.
\end{theorem}
\begin{proof}
Since each label $(v, dis(u, v), c)$ in $L^{SPC}$ has $dis(v, u) \leq D$, $L^{PSPC}$ contains all labels in $L^{SPC}$ and, by definition, no additional labels.
\end{proof}

\noindent \textbf{Distance Dependency.}~Definitions~\ref{def:DSLS} and~\ref{def:DPLS} provide us the option to eliminate the order dependency in the label construction process.

\begin{theorem}
\label{theo:dis_depend}
$L^{PSPC}_d(u)$ relies on $L^{SPC}_{\leq d}$. Specifically, given a node $u$, for a node $v \in V$ with $r(v) > r(u)$ and $dis(u, v) = d$, $(v, d, count(v, u)) \in L^{PSPC}_{d}(u) \in L^{SPC}_d(u)$ if and only if Query$(u, v, L^{PSPC}_{\leq d}) = (d_0, c_0)$ and $d_0 > d$
\end{theorem}
\begin{proof}
Consider a node $v$ with $dis(u, v) = d$. $S$ denotes the set of nodes on the shortest paths from $u$ to $v$ and let $w$ be the highest-ranked node in $S$. According to Theorem~\ref{theo:label_property}, there are two cases that are exclusive:
\begin{enumerate}
    \item $w = v$ iff $v$ is the hub of $i$ in $L^c$.
    \item $w \neq v$ indicates that
    \begin{enumerate}
        \item $w$ is the hub of both $u$ and $v$, and
        \item $dis(u, w), dis(w, v) \leq d$
        and hence, Query$(u, v, L^{PSPC}_{\leq d}) = (d_0, c_0)$  and $d_0 = d$.
    \end{enumerate}
\end{enumerate}
Therefore, if $(v, dis(v, u), c) \notin L^{PSCP}_{d}(u)$, namely, $v$ is not a hub of $u$, then $w \neq v$, and then $Query(u, v, L^{PSCP}_{\leq d}) = (d_0, c_0)$ and $d_0 = d$. Besides, if $(v, dis(v, u), c_0) \in L^{PSPC}_{d}(u)$, namely, $v$ is a hub of $u$, $v$ is the highest-ranked node in $S$ and therefore, no other node in $S$ can be a hub of $v$, that is, $Query(u, v, L^{PSPC}_{\leq d}) = (d_0, c_0)$ and $d_0 > d$.
\end{proof}

By transforming the order dependency to distance dependency, the index may be constructed in $D$ iterations, where $D$ denotes the diameter of the graph.

\subsection{The Parallelized Labeling Method}
\label{subsect:paralM}
To apply Theorem~\ref{theo:dis_depend} to construct $L^{PSPC}_d(u)$, it is costly to examine all the node pairs with a distance equal to $d$. This section provides a practical algorithm, Parallel Shortest Path Counting (PSPC), to construct the index $L^{PSPC}$ in label propagation.

\noindent \textbf{Propagation-Based Label Construction.}~This subsection provides a positive answer to the following question: can $L^{PSPC}_d(u)$ be built by gathering the labels of its neighbors, namely, $L^{PSPC}_{d-1}(v)$ is sufficient to create $L^{PSPC}_{d}(u)$ in Lemma~\ref{lemma:proga}.

\begin{lemma}
\label{lemma:proga}
All the hub nodes of labels in $L^{PSPC}_{d}(u)$ appear in labels $\bigcup_{v \in N(u) L^{PSPC}_{d-1}(v)}$ as hub nodes.
\end{lemma}
\begin{proof}
It is shown that if a node is not a hub of any node $v \in N(u)$ in $L^{PSPC}_{d-1}(v)$, then it is not a hub of $u$ in $L^{PSPC}_{d}(u)$. Let $w \neq u$ be a hub of $u$ in $L^{PSPC}_{d}(u)$ but is not a hub of any node $v \in N(u)$ in $L^{PSPC}_{d-1}(v)$. Note that the $SPC$ was built in a BFS search. Consider the iteration when the pruned BFS search is sourced from $w$. Since $w \neq u$ and $w$ is a hub of $u$, there is a shortest path from $w$ to $u$ such that $w$ is a hub of all nodes on the path. Let $s$ be the predecessor of $u$ on the shortest path. $s \in N(v)$ and $(w, dis(w, s), c_{w,s}) \in L^{SPC}$. Since $dis(w, s) = d - 1$, $w$ is a hub of $L^{PSCP}_{d-1}(s)$, this leads to the contradiction.
\end{proof}

\noindent \underline{\textit{Pruning Conditions}}.~According to Lemma~\ref{lemma:proga}, $L^{PSPC}(u)$ may be constructed iteratively, with the initial condition being the insertion of $u$ into the label $L^{SPC}_{0}(u)$ as its own hub. Nevertheless, pouring all nodes in $\bigcup_{v \in N(u)} L^{SPC}_{d-1}(v)$ directly into $L^{SPC}_{d}(u)$ produces a large set of candidate labels. Consequently, two rules are proposed to prune superfluous label entries.

\begin{lemma}
\label{lemma:prune1}
A hub $w$ in the label set $\bigcup_{v \in N(u)} L^{PSPC}_{d-1}(v)$ is not a hub of $u$ if $r(w) < r(u)$.
\end{lemma}

\begin{lemma}
A hub $w$ in the label set $\bigcup_{v \in N(u)}L^{PSPC}_{d-1}(v)$ is not a hub of $u$ in $L^{PSPC}_{d}(v)$ if $Query(w, u, L^{PSPC}_{\leq d}) = (d_0, c_0)$ and $d_0 < d$.
\end{lemma}
\begin{proof}
If $Query(w, u, L^{PSPC}_{\leq d}) = (d_0, c_0)$, and $d_0 < d$, then $dis(w, u) \neq d$, $w$ is not a hub of $u$ with distance $dis(w, u) = d$. If $Query(w, u, L^{PSPC}_{\leq d}) = d$, two situations are discussed:
\begin{enumerate}
    \item $dis(w, u) < d$, $w$ is not a hub of $u$ with distance $d$.
    \item $dis(w, u) = d$, there is a node $z$ on the shortest path between $w$ and $u$ with $r(z) > r(w)$. According to Theorem~\ref{theo:label_property}, $w$ is not a hub of $u$ in $L^{c}$. Then, these newly found shortest paths could be added into $L^{nc}$.
\end{enumerate}
Therefore, $w$ is not a hub of $u$ if $Query(w, u, L^{PSPC}_{\leq d}) = (d_0, c_0)$ and $d_0 < d$.
\end{proof}

Based on the above pruning rules, the label propagation function is proposed to find the exact $L^{PSPC}_{d}(u), \forall u \in V$.

$C_{d}(v)$ denote the set of hub nodes in label set $L^{PSPC}_{d}(v)$, for $\forall v \in V$ and $d \in [1, D+1]$.

\begin{definition}[\rrev{Label Propagation Function}]
\rrev{$L^{PSPC}_{d}(u) = \bigcup_{w \in C_{d-1}(v), for \forall v \in N(u)} L^{PSPC}_{d}(u, w)$
where $L^{PSPC}_{d}(u) =$}
\begin{equation}
\label{equa:PSPCdu}
\rrev{    \left\{
    \begin{aligned}
    &\emptyset \ \     if \ r(w) < r(u) \ or \\ &Query(w, u, L^{PSPC}_{\leq d}) = (d_0, c_0), \ and \ d_0 < d;                    \\
    &(w, dis(w, u)) \  \ \  otherwise      \\
    \end{aligned}
    \right
    .}
\end{equation}
\end{definition}
\begin{proof}
$L'$ denote the label set computed from Equation~\ref{equa:PSPCdu}. It is shown that $L' = L^{PSPC}_{d}(u)$ in two directions. Due to the correctness of Lemma~\ref{lemma:proga}, and the pruning conditions, the label set $L^{PSPC}_{d}(u) \subset L'$. The following parts demonstrate $L' \subset L^{PSPC}_{d}(u)$. Let $(w, dis(w,u)$ be a label in $L'$. Equation~\ref{equa:PSPCdu} shows that $r(w) > r(u)$ and $Query(w, u, L^{PSPC}_{\leq d}) \geq d$.

If in $L^{c}$, $w$ is not a hub of $u$, then according to Theorem~\ref{theo:label_property}, there exists a node $s$ that in $S$ $-$ the set of all nodes in the shortest path between $w$ and $u$ $-$ with $r(s) > r(w) > r(u)$. Therefore, $dis(w, s), dis(s, u) < d$ and $dis(w, u) \leq d$, and consequently, $Query(w, u, L^{PSPC}_{\leq d}) < d$, this leads to the contradiction.

Therefore, $w$ is a hub of $u$ in $L^{SPC}$. Additionally, if $dis(w, u) < d$, $Query(w, u, L^{PSPC}_{\leq d}) = (d_0, c_0)$ and $d_0 < d$, there is a contradiction. Thus, $dis(w, u) = d$. Now, it has been proved that $w$ is a hub of $u$ in $L^{SPC}$ with $dis(w, u) = d$, i.e., $w$ is a hub of $u$ in $L^{PSPC}_{d}(u)$ which completes the proof.
\end{proof}

\subsection{Propagation Paradigms}
\label{subsect:para}
There are two paradigms for label propagation. The first one is \textit{Push-Based} paradigm, whereas the second one is \textit{Pull-Based} paradigm. Then, the benefits and drawbacks are discussed.

\begin{definition}[Push-Based Label Propagation]
In $i^{th}$ iteration, vertex $v$ propagates its label entries to all of its out-neighbors. 
\end{definition}

The Algorithm is illustrated in the following manner:
Line~\ref{line:push_parallel} tackles all the vertices in parallel. For each vertex $v$ $\in V$, it pushes its in-neighbors in Line~\ref{line:push_in}. Followed, Line~\ref{line:insert_cans} inserts all the candidates' hubs to $\mathcal{C}(u)$. Line~\ref{line:push_remove_dup} eliminates the duplicate candidates. 
Line~\ref{line:push_cands} traverses each element in the $\mathcal{C}(u)$ recursively. Then, two pruning conditions are validated in Lines~\ref{line:push_prune1} and~\ref{line:push_prune2}, respectively. If not pruned, this index entry would be inserted into $L^{SPC}_{d}(u)$ in Line~\ref{line:push_insert_index}.

\begin{algorithm}[htb]
\vspace{-1mm}

\For{{\bf each} $u \in V$ in parallel }{
\label{line:push_parallel}
\For{{\bf each} $v_k \in G_{out}(u)$}
    {
        \label{line:push_in}
        $\mathcal{C}(v_k)$ $\leftarrow$ hubs $\in L^{SPC}_{d-1}[u]$ \;
    }
}
\For{{\bf each} $u \in V$ in parallel }{
\label{line:push_parallel2}
        Remove the duplicates in $\mathcal{C}$ \;
        \label{line:push_remove_dup}
        \For{{\bf each} $c \in \mathcal{C}(u)$}
        { 
                \label{line:push_cands}
            \If{r($c$.v) $<$ r(u)}
            {
            \label{line:push_prune1}
            continue\;
            }
            \If{Query($c.v, u, L^{SPC}_{d}$) $< $ d}
            {
                \label{line:push_prune2}
                continue\;
            }
            Insert($c$) into $L^{SPC}_{d}(u)$\;
            \label{line:push_insert_index}
        }
    
}
\caption{\textsc{PUSH}($G, L^{SPC}_{d-1}$)}
\vspace{-1mm}
\label{alg:inc-cnt}
\end{algorithm}

\begin{definition}[Pull-Based Label Propagation]
In $i^{th}$ iteration, vertex $v$ receives all the label entries from all of its in-neighbors.
\end{definition}

The details of Algorithm~\ref{alg:pull} are illustrated as follows:
Line~\ref{line:pull_parallel} tackles all the vertices in parallel. For each vertex $v$ $\in V$, it pulls its in-neighbors in Line~\ref{line:pull_in}. Followed, Line~\ref{line:insert_cans} inserts all the candidates' hubs to $\mathcal{C}(u)$. Line~\ref{line:remove_dup} removes the duplicate candidates. 
Line~\ref{line:pull_cands} traverses each element in the $\mathcal{C}(u)$ recursively. Then, in Lines~\ref{line:pull_prune1} and~\ref{line:pull_prune2}, two pruning conditions are validated, respectively. This index entry would be inserted into $L^{SPC}_{d}(u)$ in Line~\ref{line:pull_insert_index} if not pruned.

\begin{algorithm}[htb]
\vspace{-1mm}


\For{{\bf each} $u \in V$ in parallel }{
\label{line:pull_parallel}
        \For{{\bf each} $v_k \in G_{in}(v_i)$}
        {
            \label{line:pull_in}
            $\mathcal{C}(u)$ $\leftarrow$ hubs $\in L^{SPC}_{d-1}[v_k]$ \;
            \label{line:insert_cans}
        }
        Remove the duplicates in $\mathcal{C}$ \;
        \label{line:remove_dup}
        \For{{\bf each} $c \in \mathcal{C}(u)$}
        { 
        \label{line:pull_cands}
            \If{r($c$.v) $<$ r(u)}
            {
                \label{line:pull_prune1}
            continue\;
            }
            \If{Query($c.v, u, L^{SPC}_{d}$) $< $ d}
            {
                \label{line:pull_prune2}
                continue\;
            }
            Insert($c$) into $L^{SPC}_{d}(u)$\;
            \label{line:pull_insert_index}
        }
}
\caption{\textsc{PULL}($G, L^{SPC}_{d-1}$)}
\vspace{-1mm}
\label{alg:pull}
\end{algorithm}


\begin{figure*}[ht]
\centering
    \subfigure[][{\scriptsize Original Graph $G$}]{
		{\includegraphics[width=0.20\linewidth]{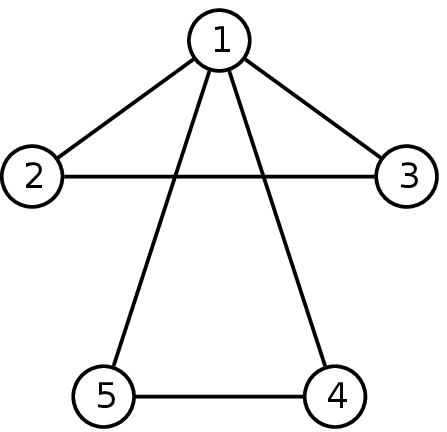}}
		\label{fig:pull_push_graph}}
    \subfigure[][{\scriptsize PULL}]{
		{\includegraphics[width=0.32\linewidth]{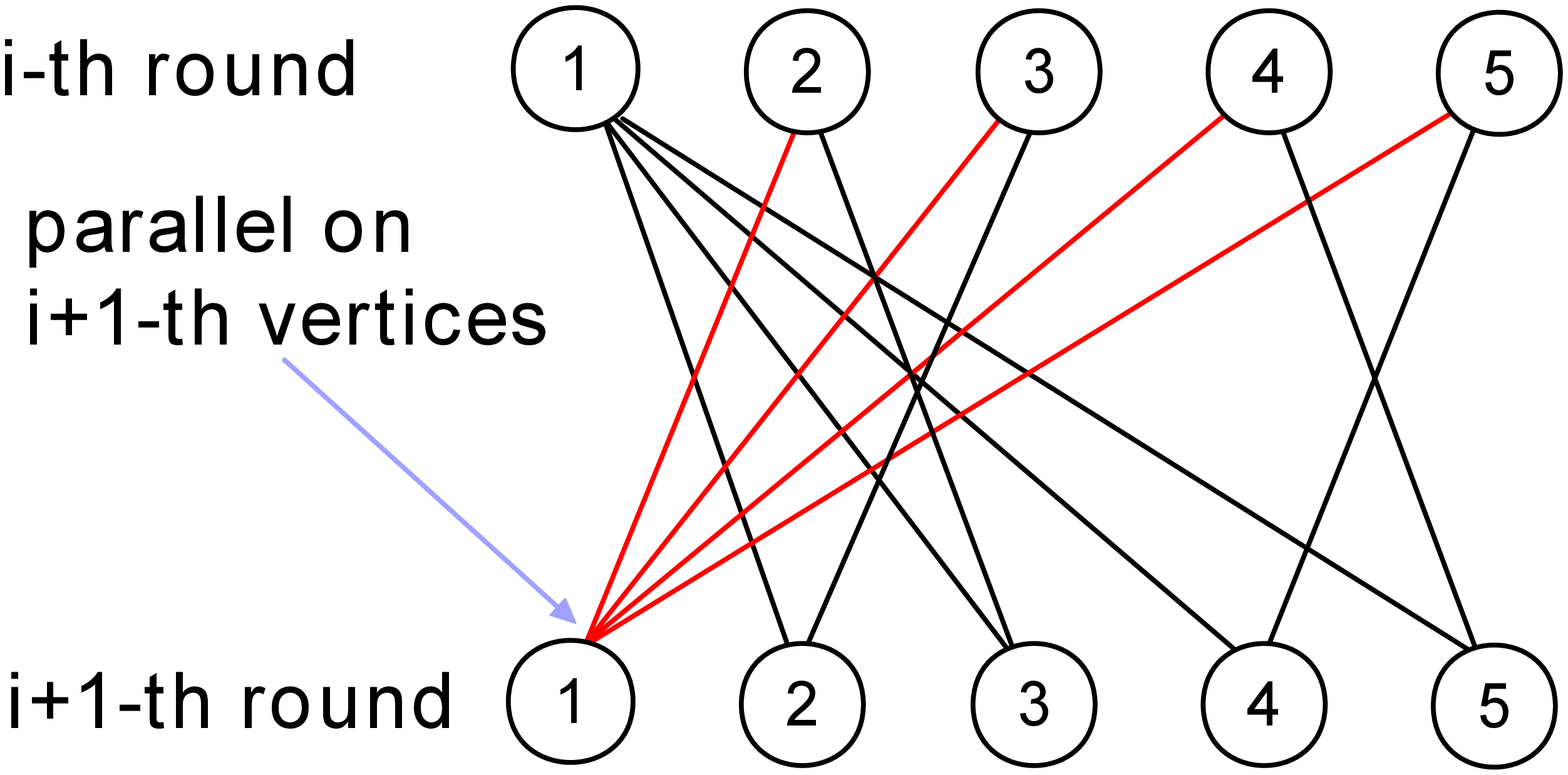}}
		\label{fig:pull}}
	\subfigure[][{\scriptsize PUSH}]{
		{\includegraphics[width=0.32\linewidth]{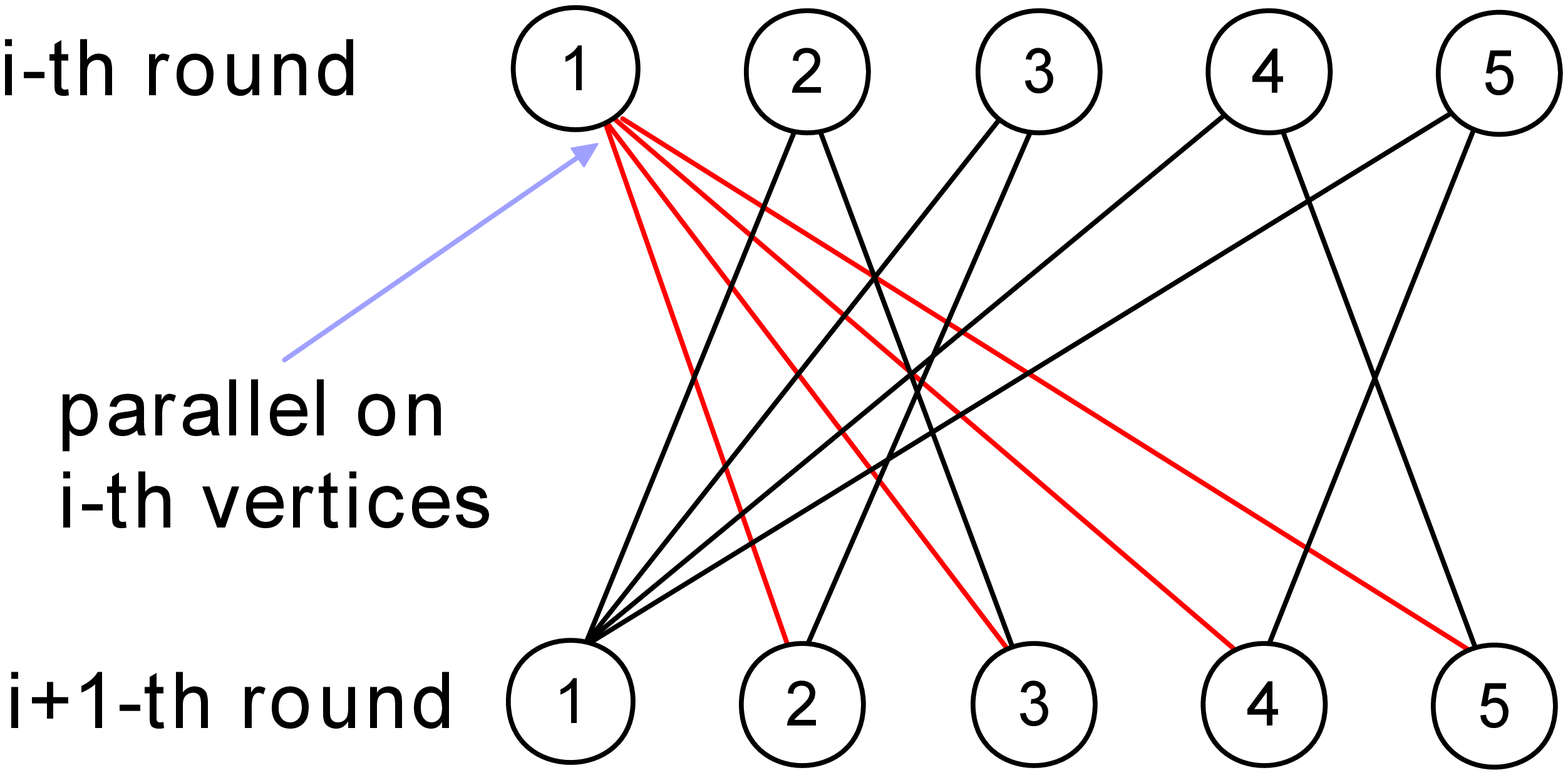}}
		\label{fig:push}}
	\caption{\small Pull-based and Push-based Paradigms.}
	\label{fig:pull_push}
\end{figure*}

\begin{example}
Figure~\ref{fig:pull_push} depicts an example of the difference between PULL-based and PUSH-based propagation paradigms. In the original graph $G$ in Figure~\ref{fig:pull_push_graph}, there are $5$ vertices and $6$ edges. In the $(i+1)$-th iteration of the index construction, Figure~\ref{fig:pull} indicates how pull-based method works. For instance, vertex $1$ received the $i$-th iteration's index entry from all of its neighbors (in-neighbors in directed graphs), i.e., vertices $2$, $3$, $4$, and $5$. If available, each vertex of the $(i+1)$-th iteration in graph~\ref{fig:pull} may be assigned with one thread. It is not necessary to allocate each edge to a single thread since there are several erroneous index entries that may be merged or removed to accelerate the process. Figure~\ref{fig:push} demonstrates how the push-based method runs. In contrast to the pull-based, the $i$-th iteration's vertex processing can be parallelize at the vertex level.
\end{example}


\noindent \textbf{Candidates Elimination.}~In this part, the duplicate removal method is briefly discussed. The reason is that there are numerous duplicate candidates for one vertex. In shortest path counting, the counting number could be fairly large. Then, one vertex $v$ would receive all of its neighbors' label entries in a single iteration. In such a scenario, there could be multiple duplicate candidates. If they were not merged, the overall computation cost would be prohibitively expensive. Thus, this part investigates the method for eliminating duplicate candidates.

The PULL-based paradigm is employed to simplify. To reduce the potential label entries for each label entry $(v, d, c)$ received by vertex $u$, there are primarily two types of operations, named \textit{Label Elimination} and \textit{Label Merging}.

\rrev{Two pruning rules are proposed to speed up this index construction process.}
\begin{itemize}
    \item \rrev{\textbf{Label Elimination.}~For two label entries $L_1$ $(v_1, d_1, c_1)$ and $L_2$  $(v_2, d_2, c_2)$ for vertex $u$, $(v_1, d_1, c_1)$, if $v_1 = v_2 \land d_1 < d_2$, then $L_2$ is eliminated by $L_1$.}

    \item \rrev{\textbf{Label Merging.}~For two label entries $L_1$ $(v_1, d_1, c_1)$ and $L_2$  $(v_2, d_2, c_2)$ for vertex $u$, $(v_1, d_1, c_1)$, if $v_1 = v_2 \land d_1 = d_2$, then $L_1$ and $L_2$ could be merged into a $L_3$ $(v_1, d_1, c_1+c_2)$.}
\end{itemize}

\begin{lemma}
    \rrev{These two pruning rules \textit{Label Elimination} and \textit{Label Merging} do not affect the correctness of the parallel shortest path counting algorithm.}
\end{lemma}
\begin{proof}
\rrev{We simply prove its correctness by contradiction. Due to the space limit, we only prove the correctness of \textit{Label Elimination}, while the proof for \textit{Label Merging} is similar. Assume Label Elimination missed a label entry $L_o$ $(v_o, d_o, c_o)$, which is caused by the elimination of $L_2$  $(v_2, d_2, c_2)$. Thus, there must exist a label entry $L'_o$ $(v_o, d_o - d_2 + d_1, c'_o)$ by replacing the $L_2$ parts with $L_1$. This contradicts the fact that $L_o$ $(v_o, d_o, c_o)$ is necessary.}
\end{proof}

\noindent\rrev{\textbf{Time and Space Complexity.}~Since these two types of operations are based on $v_1 = v_2$, $(v, d, c)$ can be stored into map by using $v$ as the key. For the label entries with the same key, it is necessary to sort them and then conduct \textit{Label Elimination} and \textit{Label Merging}. Assume there are $\eta$ candidates, the worst-case time complexity would be $O(\eta \times log \eta)$. The worst case of $\eta$ is the number of vertices whose rank is lower than the current vertex.}

\subsection{Schedule Plan}
\label{sebsect:schedule}
A basic objective of parallel index construction is to balance workloads. This subsection investigates how to allocate tasks in the most equitable manner feasible.

\noindent \underline{\textit{Node-Order-Based Schedule}}.~Assume there are $n$ vertices and $t$ threads, $order[i]$ indicates the vertex id which ranks $i$. In our index construction algorithm, tasks are allocated according to the node order. In each iteration for distance $d$, the $t_i\in [0, t - 1]$  thread would cope with all the tasks for vertex whose order ranges in $[t_i \times \lfloor n / t \rfloor, (t_i + 1) \times \lfloor n / t \rfloor )$.

\begin{figure}[htb]
\vspace{-1mm}
	\centering
	\includegraphics[width=1\columnwidth]{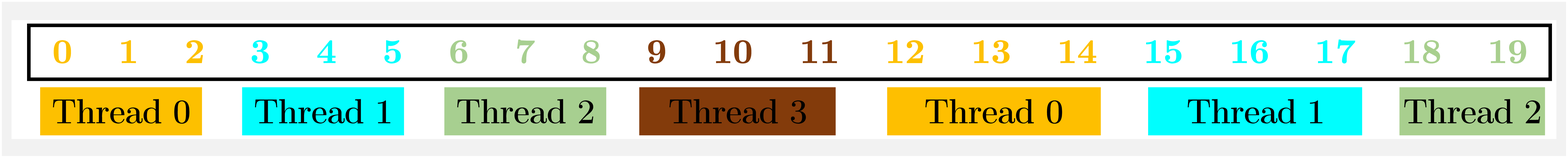}
\vspace{-1mm}
	\caption{Node-order based schedule.}
	\label{fig:node_schedule}
\vspace{-1mm}
\end{figure}

\begin{example}
\label{exm:node_order}
Figure~\ref{fig:node_schedule} illustrates an example of how the node-order based schedule works. There are $20$ vertices and $7$ threads. Consequently, threads $0$ would cope with vertices with order values in $(0, 1, 2)$. Such a schedule plan would result in an imbalance between vertices. For instance, in \pull, vertices with order $0$ would receive no candidates according to Lemma~\ref{lemma:prune1}.
\end{example}

\noindent \underline{\textit{Cost-Function-Based Dynamic Schedule}}.~Motivated by Example~\ref{exm:node_order}, such a schedule could cause an imbalance of tasks for threads. Therefore, this part proposes a dynamic cost-function-based schedule for improved scalability.

\begin{definition}
\label{def:cost_function}
$cost(v_i) = \sum_{v_j \in Nbr(v_i)} \{ |u| | u \in L_{v_j} \\
\ and \ r(v(u)) < r(v_i) \} $
\end{definition}

Since the precise cost function is expensive to calculate, an approximation-based approach is proposed for measuring the cost and the schedule.
In addition, rather than statically assigning each thread with an equal amount of tasks, it dynamically allocates tasks when tasks on it finish.



\subsection{Vertex Ordering Strategies}
The vertex order is crucial for both \base \ and \our \ since it has a significant impact on indexing time, index size, and the query time. An efficient ordering scheme should rank vertices that cover more shortest paths higher s.t. later searches in \base \ can be reduced as early as possible, thereby reducing the number of label entries generated. 
In the literature, several heuristics for obtaining such orderings have been investigated. For the sake of completeness, two state-of-the-art schemes are reviewed, namely, degree-based and significant-path-based, below.

\noindent \underline{\textit{Degree-Based Scheme}}.~Vertices are arranged in ascending degree order.
This technique is based on the premise that vertices with a higher degree have stronger connections to many other vertices, and as a result, many shortest paths will pass through them.

\noindent \underline{\textit{Significant-Path-Based Scheme}}.~The significant-path-based method is more adaptable than the degree-based scheme, which utilizes only local information. Let $w_1$, $w_2$, ..., $w_n$ be the ordering generated under this scheme, where $w_i$ is the $i$-th hub to be pushed. Given $w_i$, the scheme determines $w_{i+1}$ as follows. When pushing hub $w_i$ in SPC, a partial shortest-path tree $T_{w_i}$ rooted at $w_i$ will be produced. For each vertex $v$ in $T_{w_{i}}$, let $des(v)$ be the number of descendants and par$(v)$ be the parent of $v$. Beginning with $w_i$, the scheme computes a significant path $P_{sig}$ to a leaf by iteratively selecting a child $v$ with the largest des($v$). $p_{sig}$ is intuitively a path that many shortest paths intersect. Then, among all vertices on $p_{sig}$ other than $w_i$, vertex $v$ with the largest deg($v$) $\cdot$ $(des(par(v)) - des(v))$ is empirically selected as $w_{i+1}$. $w_1$ is initially configured to have the highest degree in $G$. The Significant-Path-Based Scheme is the most efficient order in~\cite{zhang2020hub}. Nevertheless, as stated before, this ordering needs to compute the next vertex based on the current vertex's shortest path tree, which naturally includes a dependency on the index construction process and hence is not suited for parallel processing.





Halin~\cite{halin1976s}, and Robertson et al.~\cite{robertson1984graph} introduced tree decomposition, which is a technique for mapping a graph to a tree in order to expedite the resolution of certain computational problems in graphs.
\eat{Numerous NP-complete algorithmic problems, such as maximum independent set and Hamiltonian circuits, for arbitrary graphs, may be efficiently solved by dynamic programming for graphs of bounded treewidth, employing the tree-decompositions of these graphs.} Bodlaender's introductory survey can be found in~\cite{bodlaender1994tourist}. Vertices are naturally arranged in a hierarchy through tree decomposition. Therefore, this paper utilizes tree decomposition to create the vertex hierarchy, and demonstrate that the effectiveness of this hierarchy in answering shortest path counting queries in a road network. 
Given a graph $G(V, E)$, a tree decomposition of it is defined as follows~\cite{bodlaender1994tourist}:
\begin{definition}[Tree Decomposition] A tree decomposition of a graph $G(V,E)$, denoted as $T_G$, is a rooted tree in which each node $X \in V(T_G)$ is a subset of $V(G)$ (i.e., $X \subset V(G)$) such that the following three conditions hold:
\begin{itemize}
    \item $\bigcup_{X \in V(T_G)} X = V$;
    \item For every $(u,v) \in E(G)$, there exists $X \in V(T_G)$ s.t. $u \in X$ and $v \in X$.
    \item For every $v \in V(G)$ the set $\{ X | v \in X \}$ forms a connected subtree of $T_G$.
\end{itemize}
\end{definition}

\noindent \underline{\textit{Road Network Order}}.~In the road networks, the degree order is inefficient since there are two many low-degree vertices with the same degree. Thus, a Tree Decomposition method is proposed in~\cite{ouyang2018hierarchy} to address the shortest distance queries. For instance, there are $v_1,...,v_n$ vertices and a vertex order is demanded. This technique also has a natural dependency structure. The main steps of obtaining this order could be listed as follows:
\begin{itemize}
    \item Set $Q = \emptyset$ as a queue. In this first iteration, vertex $u_0$ with the lowest degree would be inserted into $Q$. Then, $u_0$ is removed from this graph, for vertex $u \in Neighbor(u_0)$, this step connects them and update their degree as $deg(u) + deg(u_0) - 1$. 
    \item Likewise, in the $i$-th iteration, the $u_i$ (the lowest degree vertex in the current graph) are pushed into the $Q$. Then, $u_i$ is removed from this graph, for vertex $u \in Neighbor(u_i)$, this step connects them and update their degree as $deg(u) + deg(u_i) - 1$.
    \item This step could produce a resultant vertex order by append vertices in $Q$ into the $R$ from the back of the queue to the front.
\end{itemize}

Then, the original graph could be divided into core- and fringe-parts, and employ a hybrid vertex ordering.

\underline{\textit{Hybrid Vertex Ordering}}.~Therefore, a hybrid vertex ordering is proposed, which compromises between the computational efficiency of degree vertex order and the index size effectiveness of the road network order. All vertices are divided into two categories, core-part, and fringe-part. To achieve this, a threshold $\delta$ is set for the node degree. If a vertex $v$, has a degree larger than $\delta$, then $v$ would be divided into core-part. Otherwise, it would be divided into the fringe-part. Each part would utilize the corresponding order.


\subsection{Landmark-Based Filtering}
\label{subsect:land_prune}
This subsection introduces landmark pruning. Landmark-based techniques are proved to be efficient in the Label Constrained Reachability problem~\cite{valstar2017landmark}. The overall idea is to choose small groupings of vertices as landmarks and construct some small indexes on them to accelerate the whole index construction process. This work follows the strategies in~\cite{valstar2017landmark} to select landmarks. To begin, the definition of landmarks is introduced.
\begin{definition}
\label{def:land}
(Landmark).~
Given a threshold $\theta$, a vertex $v$ is a landmark if and only if $degree(v) \geq \theta$.
\end{definition}

In the index construction process, the distances from landmarks are frequently used. Based on this observation, this work proposes a landmark-based filtering to further expedite the index construction process. 

In the classical shortest path counting with a 2-hop labeling index, vertices are explored individually. It processes vertex serially. For a vertex $u$, once it is explored, there is no need to utilize its distance information due to the pruning rule 1. Thus, there is no need to construct such a landmark index. 

Nevertheless, the facts alter for our parallel algorithm. Regarding the parallel algorithm, the label propagation is divided by distance iteration. Since the degree of the landmarks is quite high, the labels from the landmarks would be the majority in each iteration. Consequently, if the distance information is stored from these landmarks, these queries could be instantly answered if landmarks are involved.

Furthermore, since all the distances are in increasing order. One bit is needed to store the ``True'' or ``False'' information, and it is sufficient to answer queries for pruning.


\section{Index Size Reduction}
\label{sect:indexReduction}

Two index reduction techniques  proposed in~\cite{zhang2020hub} are analyzed in the context of parallel. They are $1$-shell reduction and neighborhood-equivalence reduction to reduce the index size.

\subsection{Reduction by $1$-Shell}
\label{subsect:1_shell}
Shortest path counting problem is simple in the tree: there is only one shortest path between any two vertices in it. Despite the fact that graphs are often far more complex than trees, trees are more common. Graph $G$ can always be transformed into a core-fringe structure. If it is not empty, it consists of trees, as will be shown following. Additionally, each of these trees connects to the remaining of $G$ with no more than one edge. Consequently, it is safe to trim $G's$ fringe without affecting the shortest paths inside the core. In this situation, the $1$-shell of $G$ is characterized by the fringe. Specifically, the $1$-shell of $G$ is defined as the maximum subgraph within which each vertex is incident to at least $k$ edges. The $k$-core of $G$ is defined as the largest subgraph within which each vertex is incident to at least $k$ edges.

During the index construction process of our parallel algorithm, the graph could be divided into a core-fringe structure. Regarding the fringe structure consisting of trees, they can be deleted from the original graph and utilize them during the query procedure. Therefore, the Reduction by $1$-Shell technique does not influence our parallel paradigm.

\noindent \textbf{Query Evaluation in Parallel.}~A query$(s,t)$ is processed in the following manner. If $shr(s) = shr(t)$, $1$ is directly returned; otherwise, a query$(shr(s),shr(t))$ is issued on $G(s)$ and its result is returned.

\subsection{Reduction by Equivalence Relation}
\label{subsect:reduce_er}
Given an undirected graph $G$ and two vertices $u$ and $v$, it is straightforward to demonstrate that $spc_G(u, w)$ $=$ $spc_G(v, w)$ for $\forall w \neq u, v$ if $u$ and $v$ share the same neighborhoods. In fact, a shortest path from $u$ to $w$ can be modified with a shortest path from $v$ to $w$ by substituting $u$ for $v$, and vice versa. In this situation, the label of $v$ alone would serve to answer any query involving $\{u, v\}$, and other vertices. In other words, it is permissible to eliminate the label of $u$, hence reducing the size of the index.
$u$ is neighborhood equivalent to $v$, indicated by $u \equiv v$, if $nbr(u)\ \{v\} = nbr(v) \ \{u\}$. In other words, if $u$ and $v$ are not adjacent, they must have the same set of neighbors; otherwise, the neighbors of $u$ and $v$ must be identical after eliminating $u$ and $v$.
It is feasible to demonstrate that is an equivalence relation. This equivalence has been implemented in graph reduction tasks such as subgraph isomorphism~\cite{delling2014robust} and distance hub labeling~\cite{fan2012query}. As will be shown, the relation may also be utilized to reduce the graph size. However, straight application without adjustment might result in findings that are grossly underestimated.

During the index creation phase of our parallel approach, vertices with an Equivalence Relation are eliminated, leaving a single vertex to represent them. A weight is assigned to it depending on the quantity of equivalents.

\noindent \rrev{\textbf{Query Evaluation.}~Next, two query schemes are discussed for handling a query $(s, t)$. Assume w.l.o.g. that $s,t \in I$ and $s \neq t$. In this case, $Rs = nbr(s)$ and $Rt = nbr(t)$.}

\rrev{As for the query process, it could also speed up with a parallel framework. It could be parallel twofold: i) Since each query is independent of the other, it is natural to dynamically assign the query to the available thread. ii) Since the query is a set intersection, it could parallelize at the label entry granularity. Then, all the counting of the shortest path is summed up.}

\section{Experimental Results}
\label{sect:exp}

\begin{figure*}
    \centering
    {\includegraphics[width=1\linewidth]{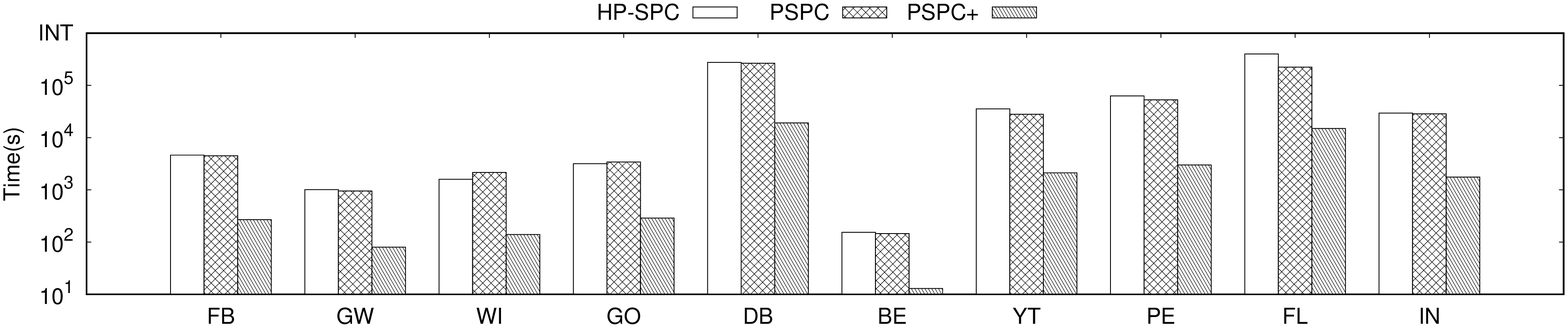}}
    \caption{Indexing Time (s) for \base, \our, and \ourp.}
    \label{fig:runtime}
\end{figure*}

\begin{figure*}
    \centering
    {\includegraphics[width=1\linewidth]{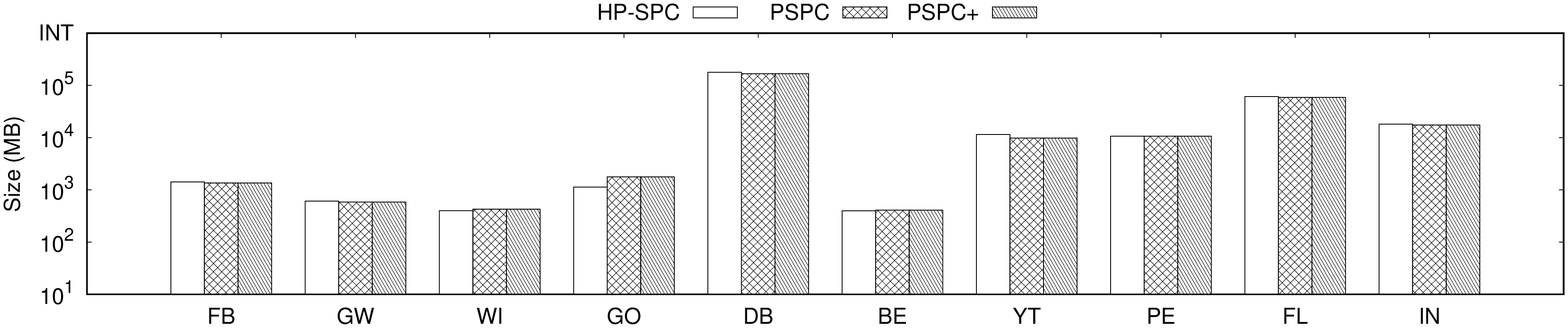}}
    \caption{Indexing Size (MB) for \base, \our, and \ourp.}
    \label{fig:indexsize}
\end{figure*}

\begin{figure*}
    \centering
    {\includegraphics[width=1\linewidth]{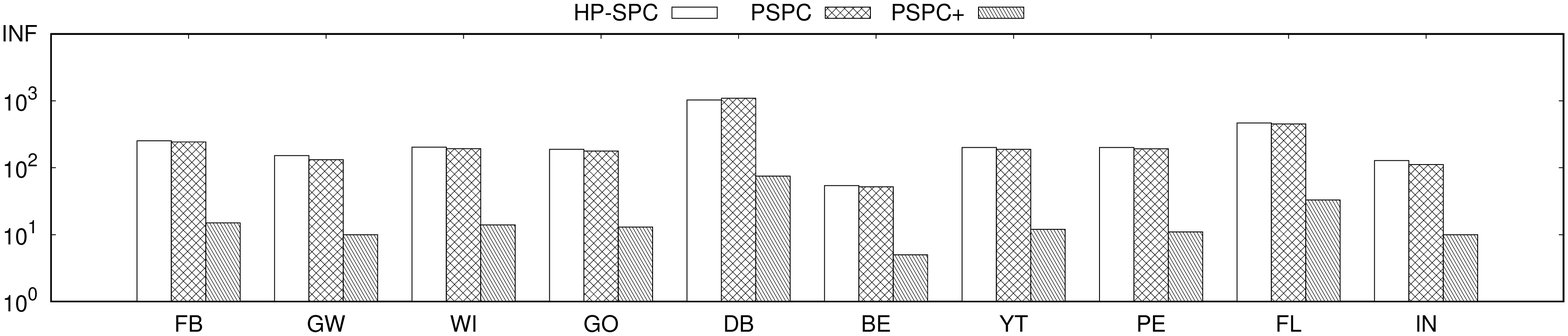}}
    \caption{Query Time (microsecond) for \base, \our, and \ourp.}
    \label{fig:querytime}
\end{figure*}

\begin{figure*}[htb]
	\vspace{-2mm}
	\newskip\subfigtoppskip \subfigtopskip = -0.1cm
	\newskip\subfigcapskip \subfigcapskip = -0.1cm
	\centering
     \subfigure[\small{FB}]{
    \includegraphics[width=0.22\linewidth]{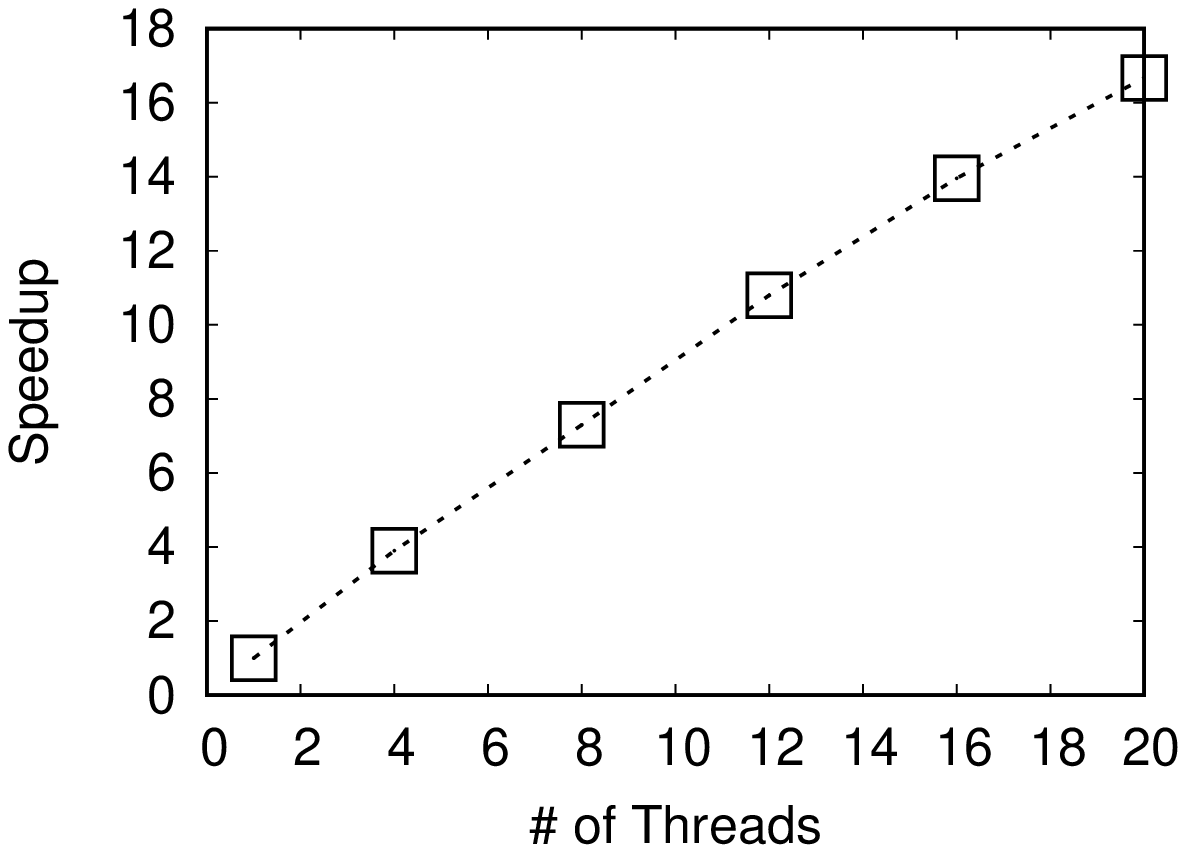}
    \label{fig:fb}
     }
     \subfigure[\small{GO}]{
    \includegraphics[width=0.22\linewidth]{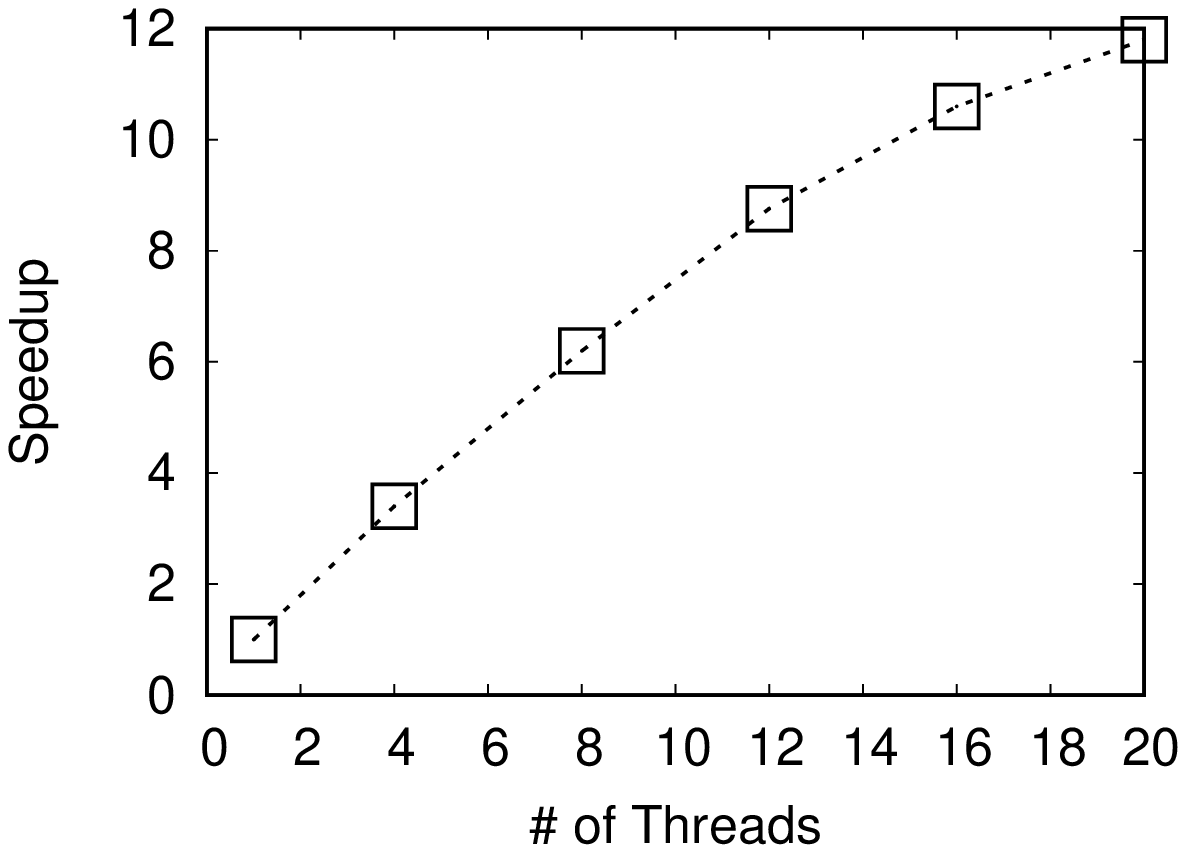}
    \label{fig:go}
     }
    \subfigure[\small{GW}]{
    \includegraphics[width=0.22\linewidth]{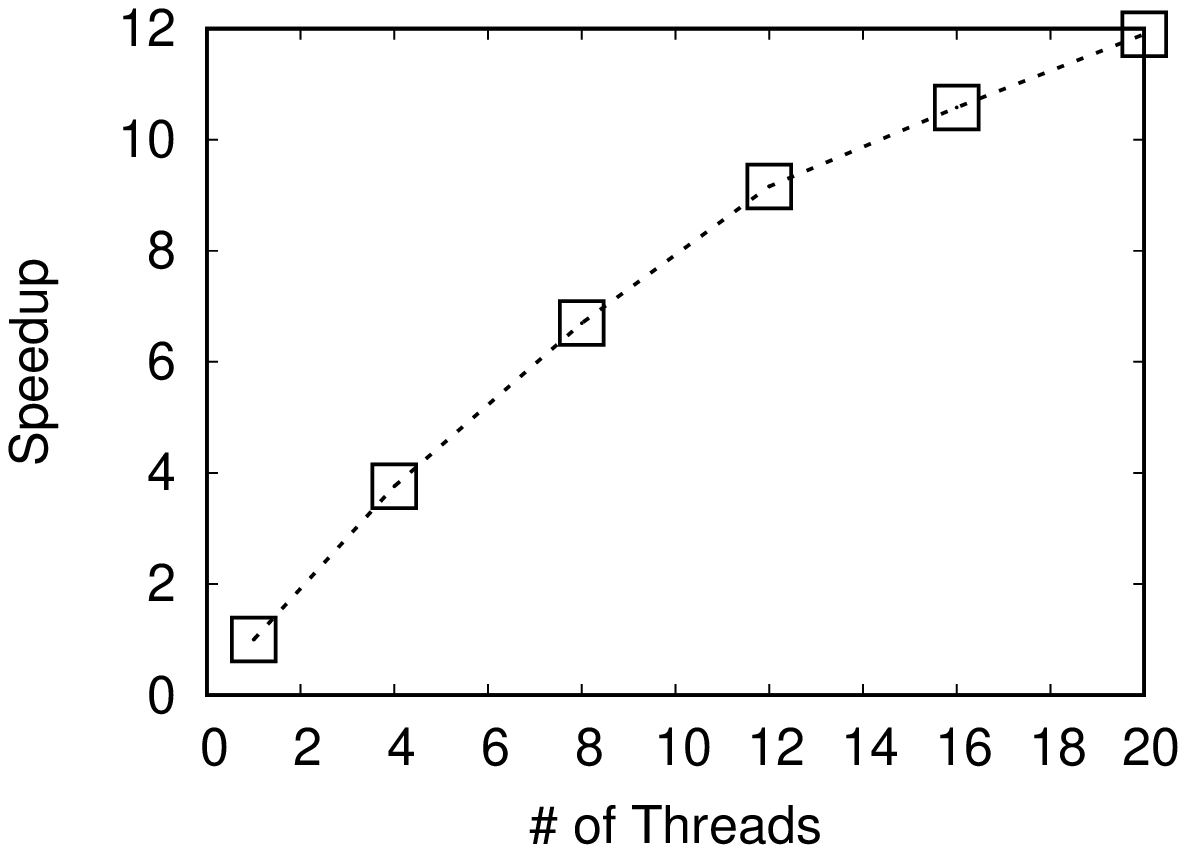}
    \label{fig:gw}
     }
     \subfigure[\small{WI}]{
    \includegraphics[width=0.22\linewidth]{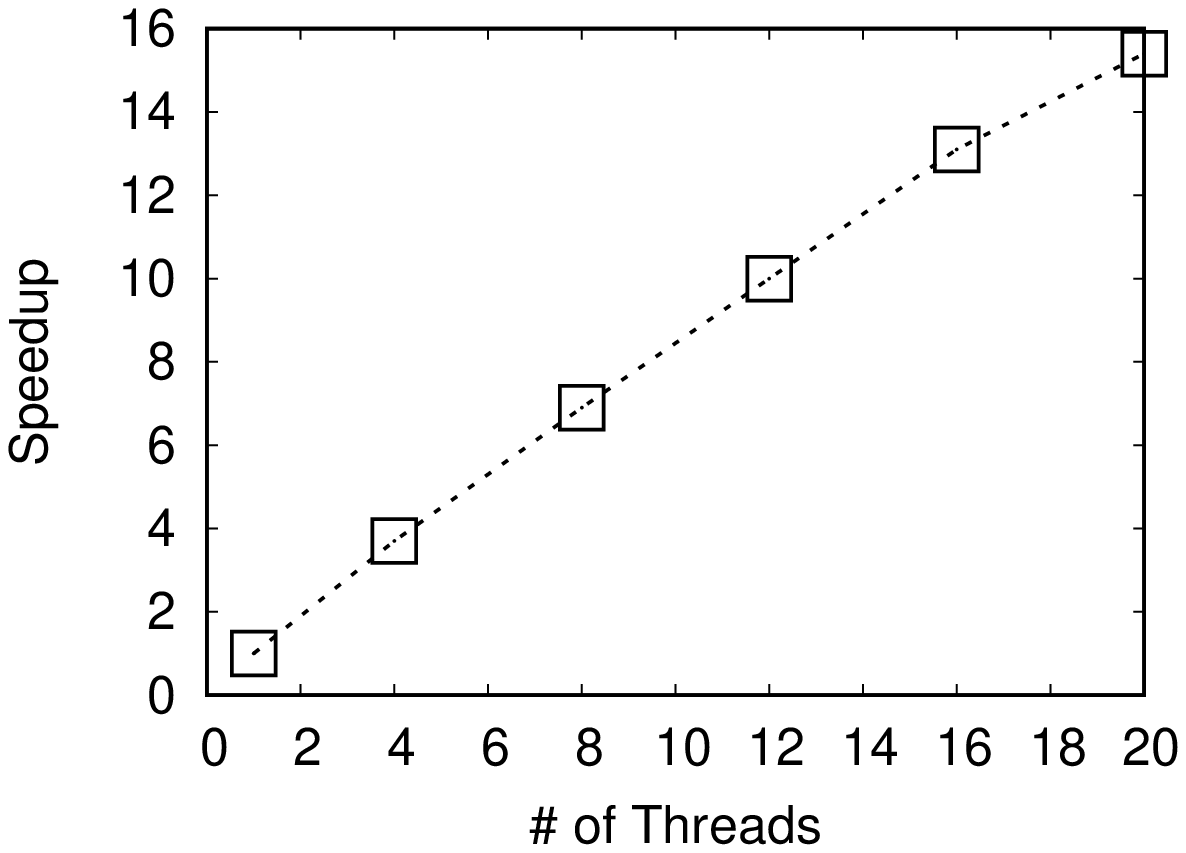}
    \label{fig:wi}
     }
\caption{Speedup of indexing time when tuning the $\#$ of threads.}
\label{fig:speedup_index_time}
\end{figure*}

\begin{figure*}[htb]
	\vspace{-2mm}
	\newskip\subfigtoppskip \subfigtopskip = -0.1cm
	\newskip\subfigcapskip \subfigcapskip = -0.1cm
	\centering
     \subfigure[\small{FB}]{
    \includegraphics[width=0.22\linewidth]{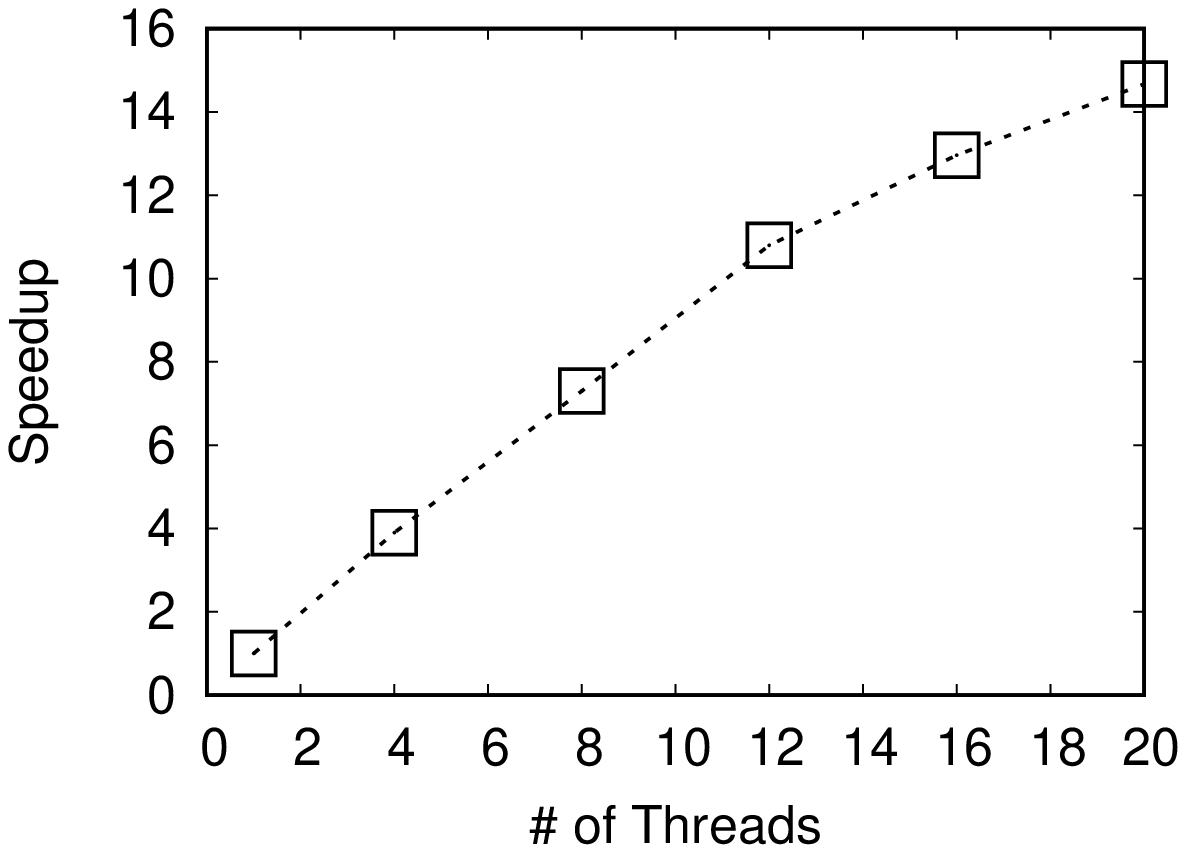}
    \label{fig:fb_qt}
     }
     \subfigure[\small{GO}]{
    \includegraphics[width=0.22\linewidth]{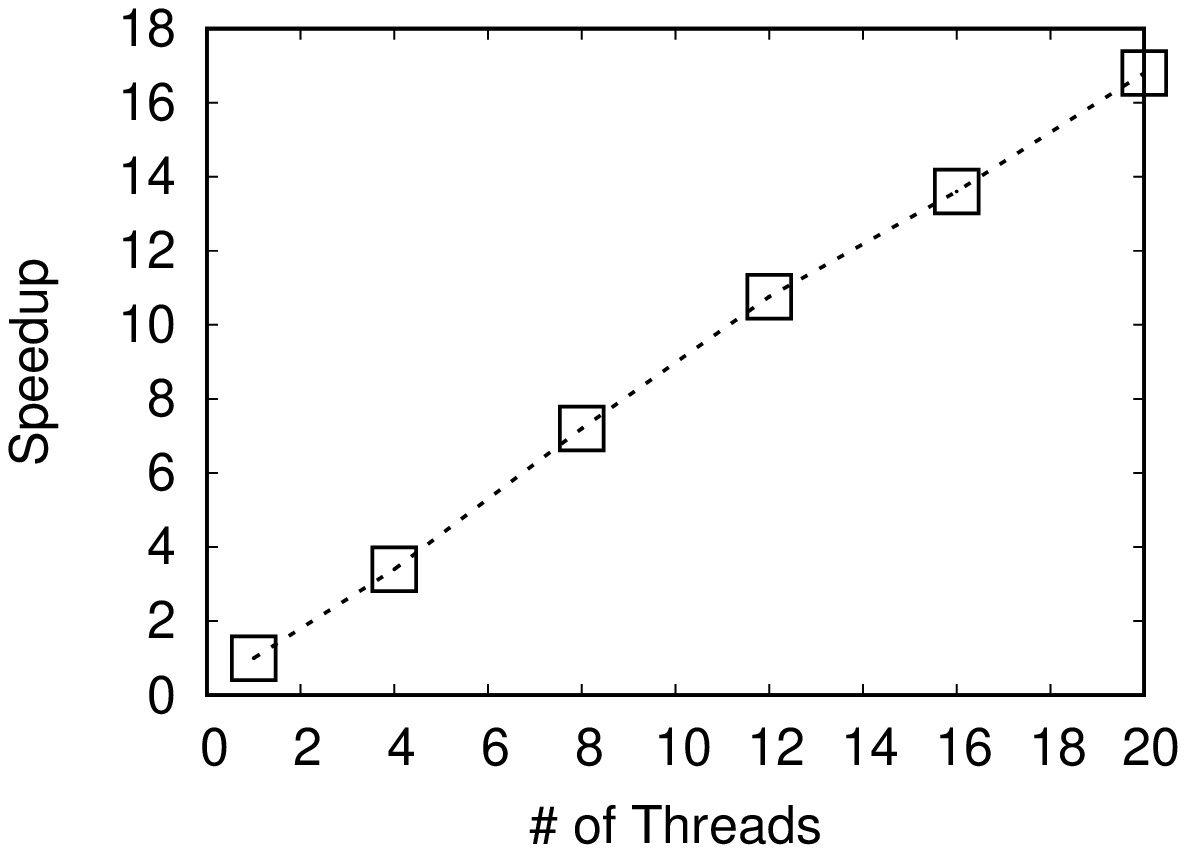}
    \label{fig:go_qt}
     }
    \subfigure[\small{GW}]{
    \includegraphics[width=0.22\linewidth]{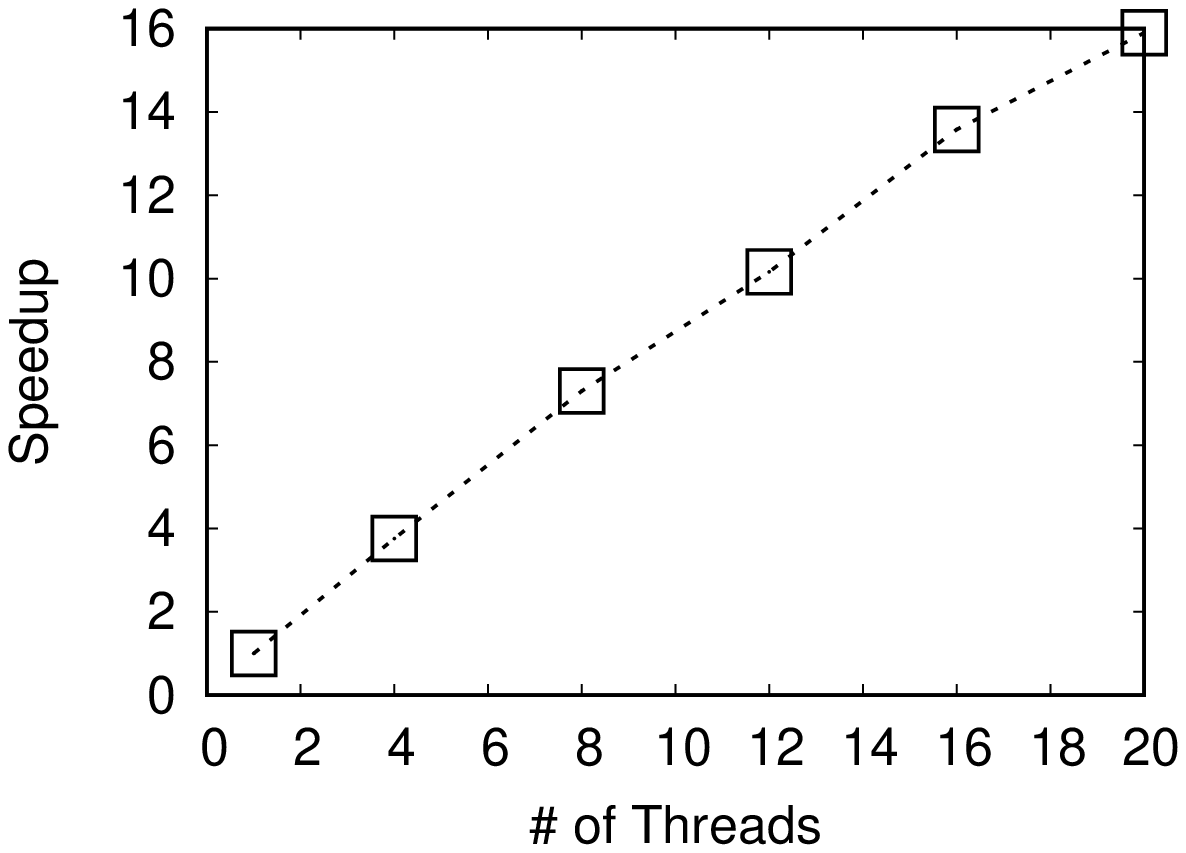}
    \label{fig:gw_qt}
     }
     \subfigure[\small{WI}]{
    \includegraphics[width=0.22\linewidth]{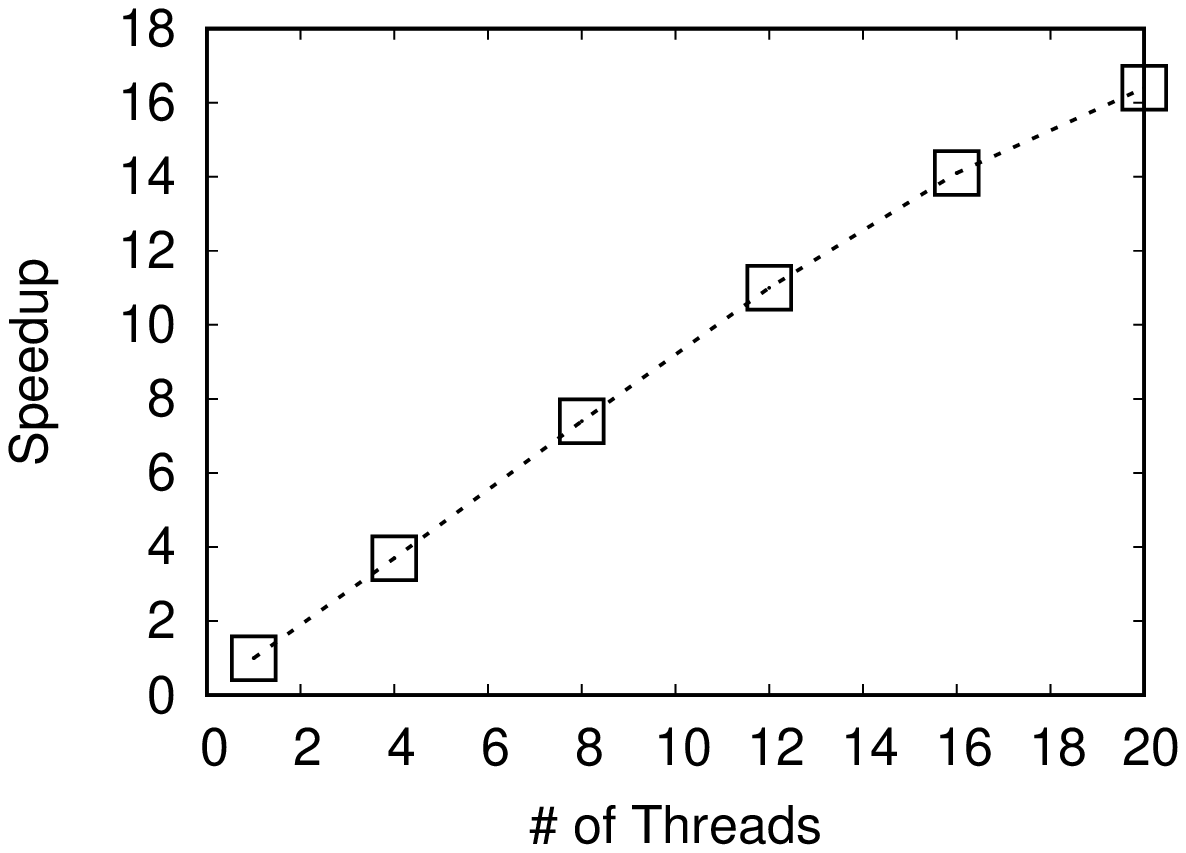}
    \label{fig:wi_qt}
     }
\caption{Speedup of query time when tuning the $\#$ of threads.}
\label{fig:speedup_query_time}
\end{figure*}

\begin{figure}[htb]
	\vspace{-2mm}
	\newskip\subfigtoppskip \subfigtopskip = -0.1cm
	\newskip\subfigcapskip \subfigcapskip = -0.1cm
	\centering
     \subfigure[\small{Landmark labeling}]{
    \includegraphics[width=0.46\linewidth]{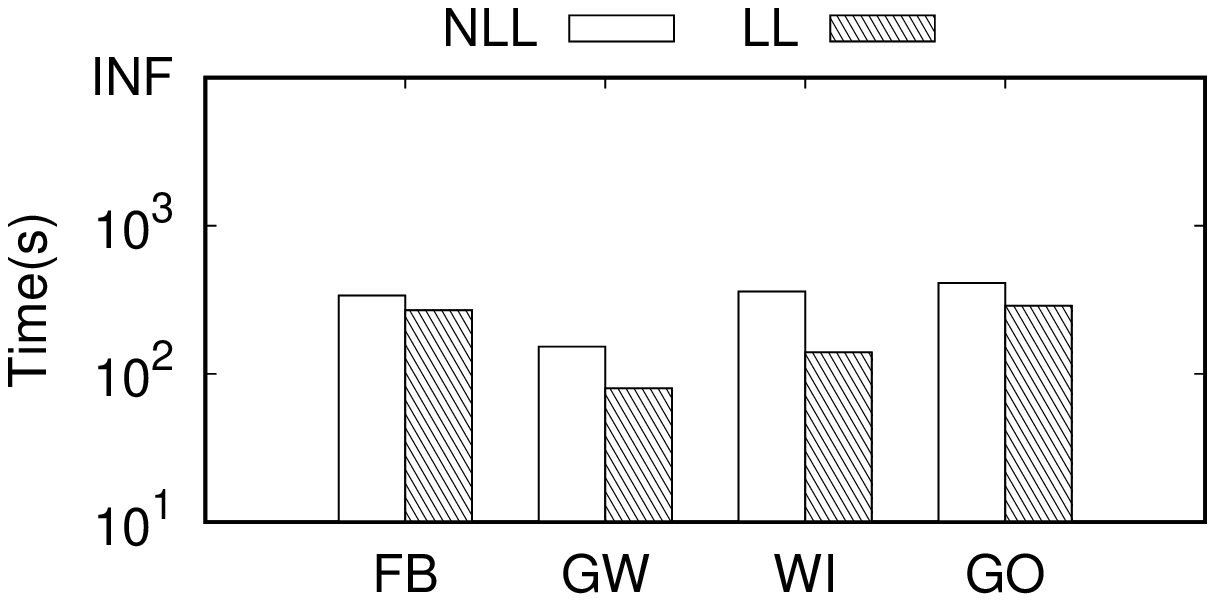}
    \label{fig:landmark}
     }
     \subfigure[\small{Schedule plan}]{
    \includegraphics[width=0.46\linewidth]{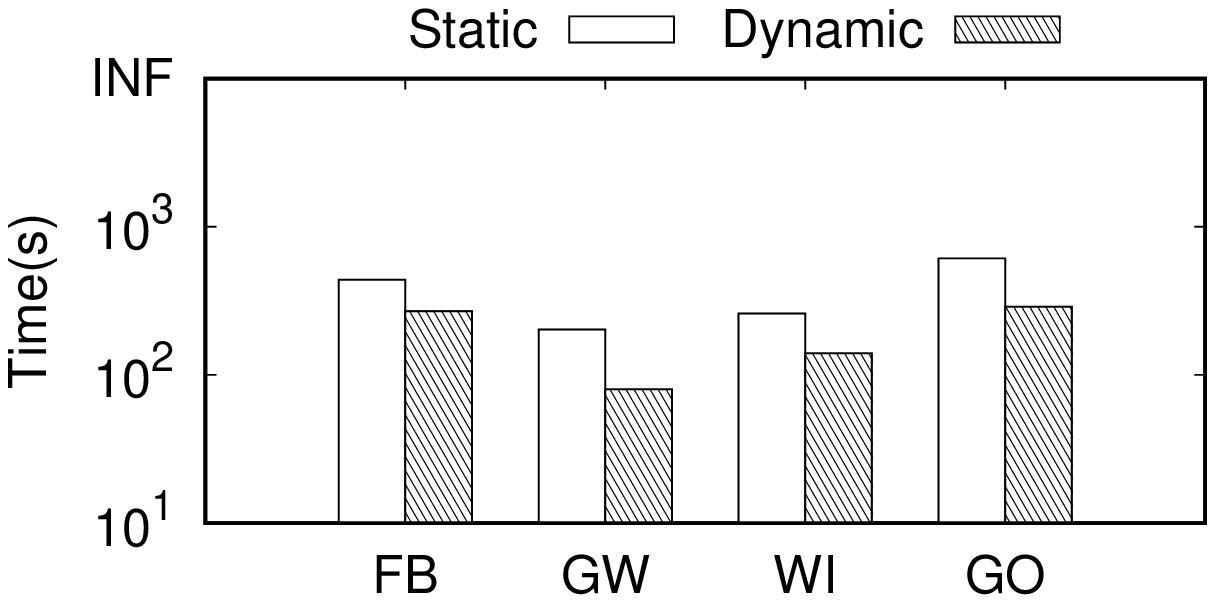}
    \label{fig:schedule}
     }
   
     \subfigure[\small{Node order}]{
    \includegraphics[width=0.96\linewidth]{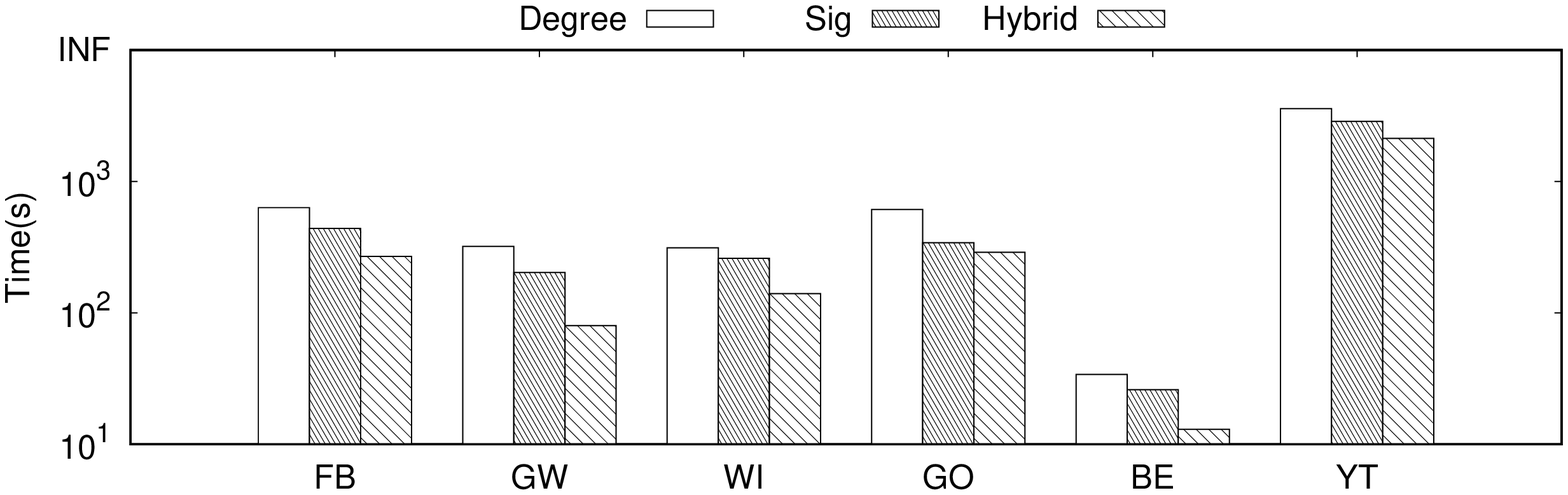}
    \label{fig:nodeorder}
     }
\caption{Ablation analysis of different techniques with $20$ threads.}
\label{fig:Ablation}
\end{figure}

\begin{figure*}[htbp]
	\vspace{-2mm}
	\newskip\subfigtoppskip \subfigtopskip = -0.1cm
	\newskip\subfigcapskip \subfigcapskip = -0.1cm
		\includegraphics[width=0.8\linewidth]{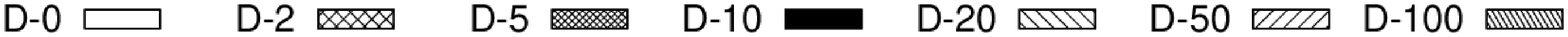}
	\centering
     \subfigure[\small{Index Size(MB)}]{
    \includegraphics[width=0.30\linewidth]{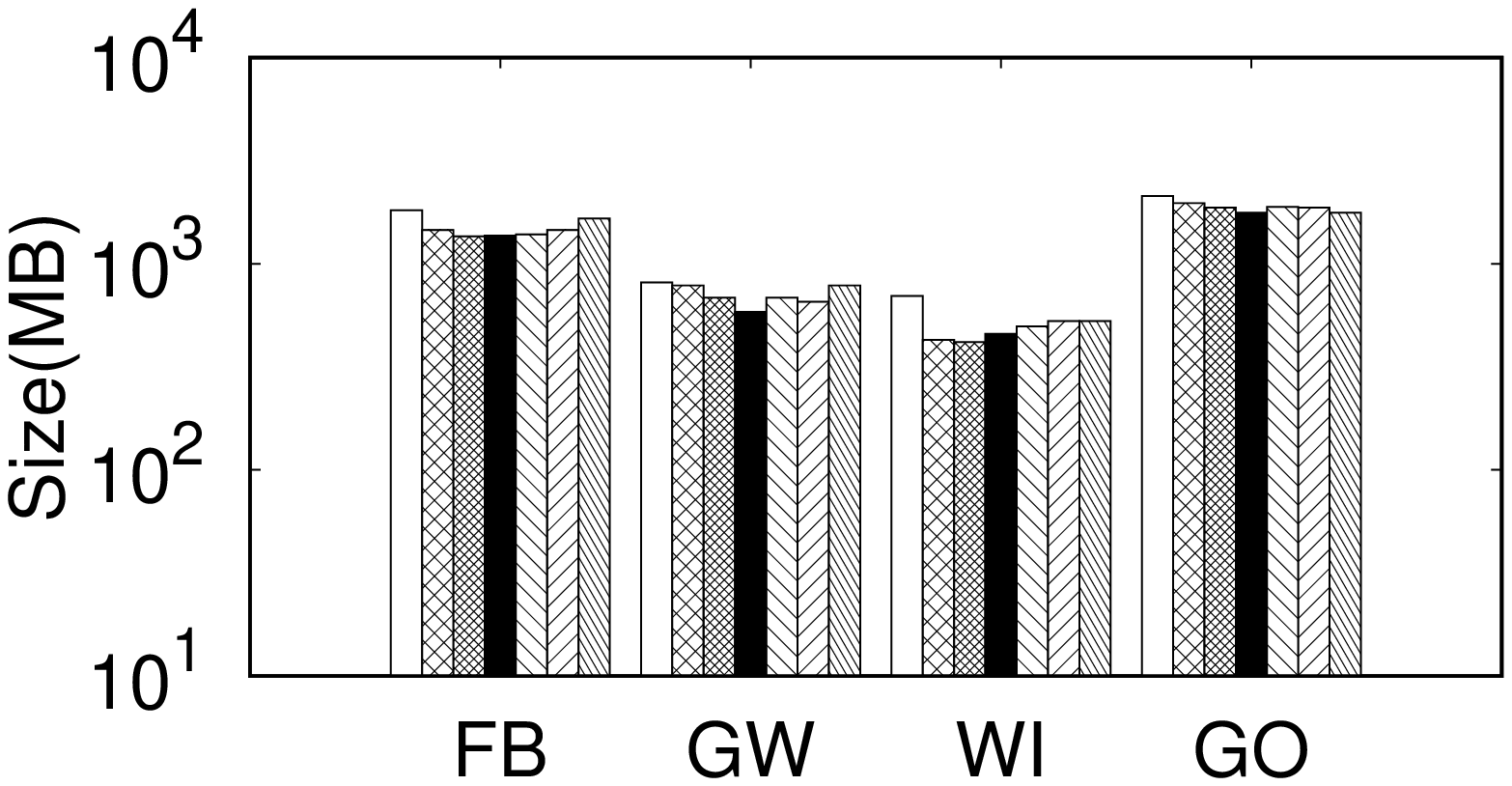}
    \label{fig:fb}
     }
     \subfigure[\small{Index Time(s)}]{
    \includegraphics[width=0.30\linewidth]{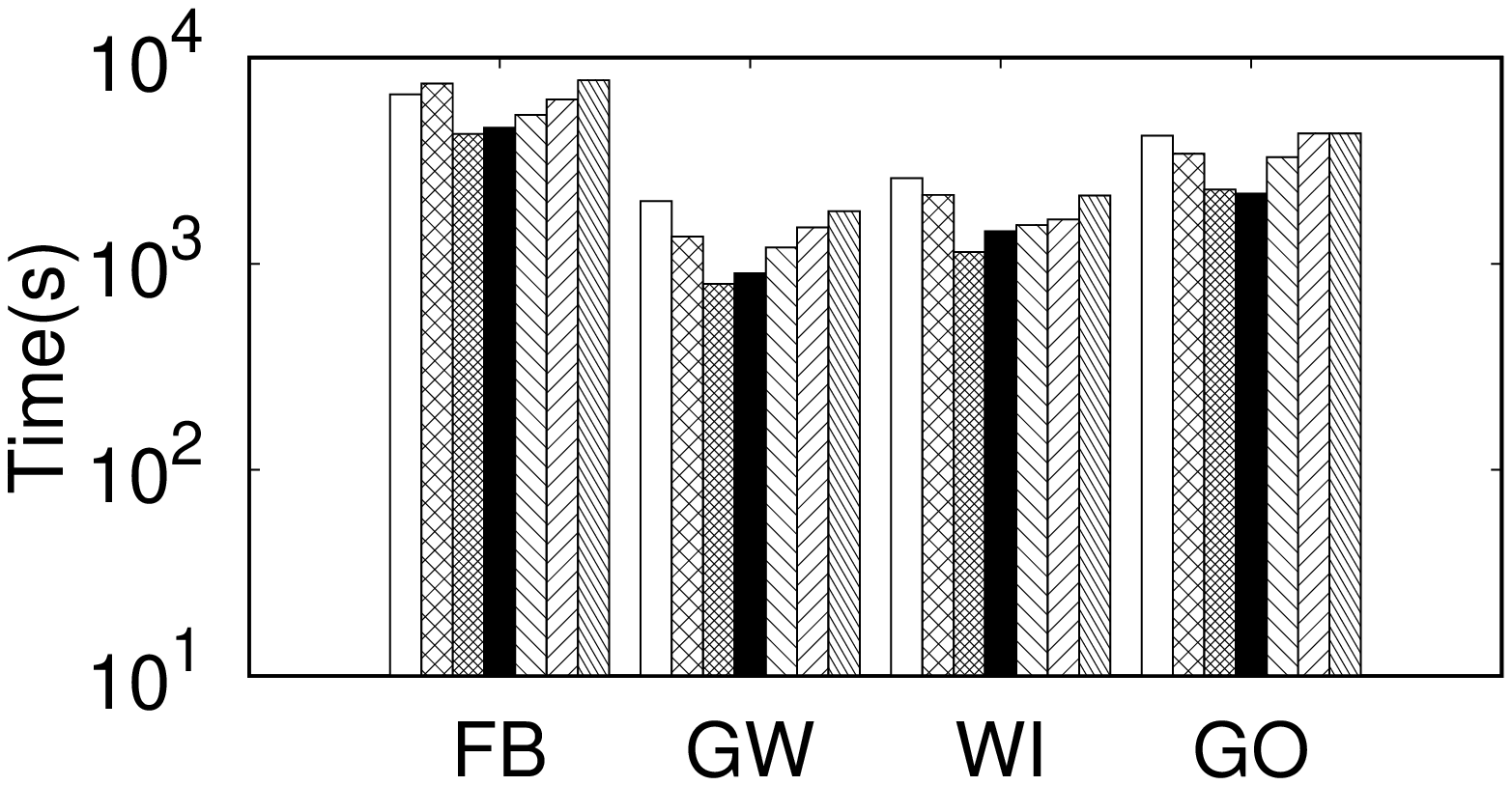}
    \label{fig:go}
     }
    \subfigure[\small{Query Time(us)}]{
    \includegraphics[width=0.30\linewidth]{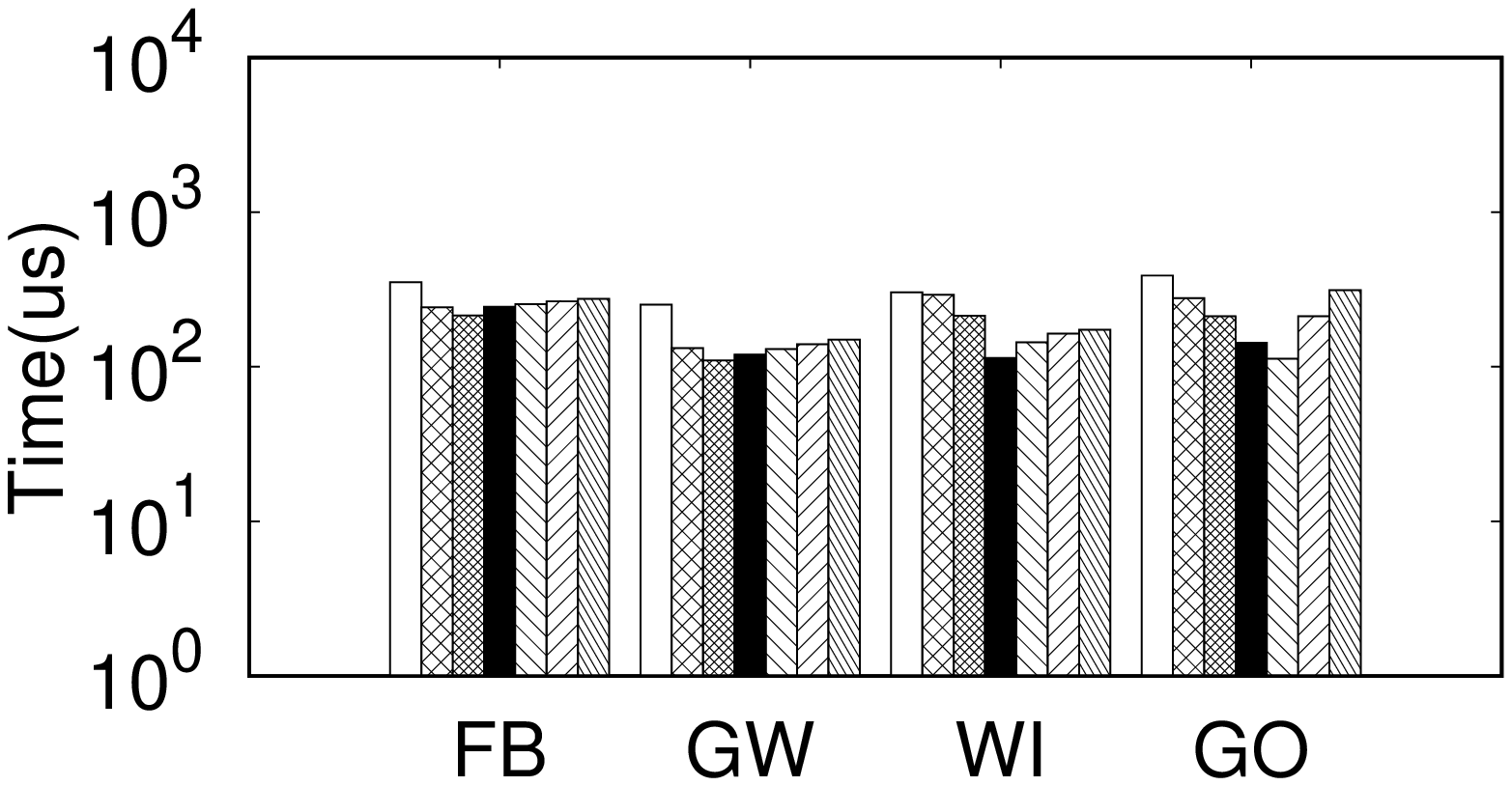}
    \label{fig:gw}
     }
\caption{The effect of threshold $\delta$.}
\label{fig:effect_delta}
\end{figure*}

  \begin{figure}[htb]
	\vspace{-2mm}
	\newskip\subfigtoppskip \subfigtopskip = -0.1cm
	\newskip\subfigcapskip \subfigcapskip = -0.1cm
	\centering
   
     \subfigure[\small{Node order}]{
    \includegraphics[width=0.96\linewidth]{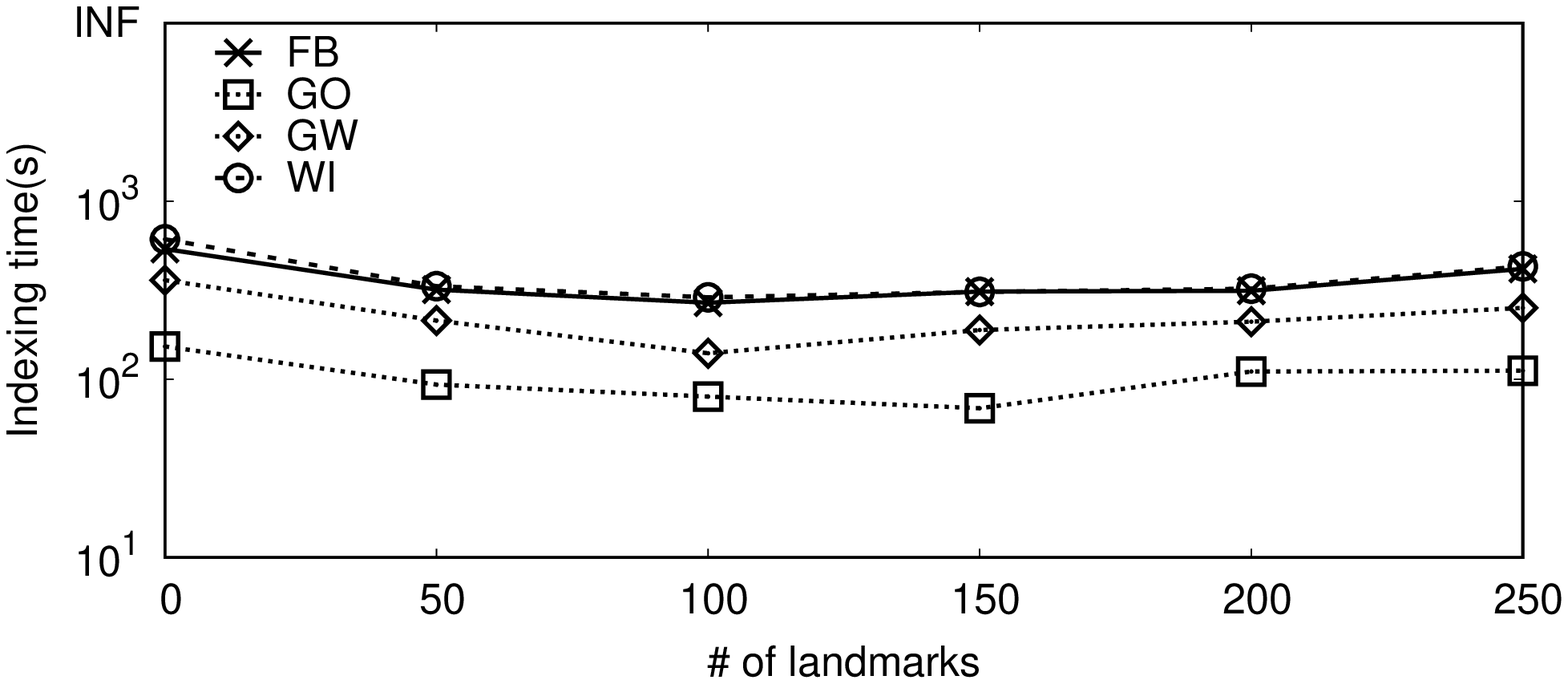}
     }
\caption{The effect of $\#$ of landmarks.}
\label{fig:effect_landmarks}
\end{figure}

  \begin{figure}[htb]
	\vspace{-2mm}
	\newskip\subfigtoppskip \subfigtopskip = -0.1cm
	\newskip\subfigcapskip \subfigcapskip = -0.1cm
	\centering
   
    {
    \includegraphics[width=0.96\linewidth]{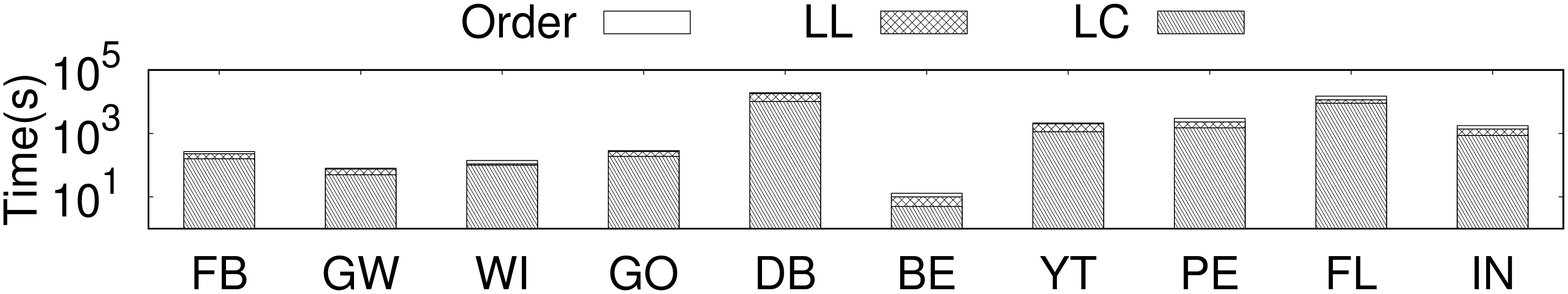}
     }
\caption{The time cost of different part during index construction. Order indicates the time cost for node ordering, while LL and LC denote the Landmark labeling and Label Construction, respectively.}
\label{fig:breakDown}
\end{figure}

This section evaluates the effectiveness and efficiency of the proposed techniques on comprehensive experiments.

\subsection{Experimental Settings}
\label{sect:settings}

\vspace{1mm}
\noindent\textbf{Algorithms}.~We \rrev{compared} \rrev{the} proposed algorithms with baseline solutions.
\begin{itemize}
    \item \text{\base}.~The state-of-the-art shortest path counting algorithm proposed in~\cite{zhang2020hub}.
    \item \text{\our}.~Our parallel shortest path counting algorithm in a single thread.
    \item \text{\ourp}.~Our parallel shortest path counting algorithm in $20$ threads.
\end{itemize}

\vspace{1mm}
\noindent\textbf{Datasets}~Table~\ref{tb:datasets} shows the important statistics of real graphs used in the experiments. 10 publicly available datasets are used. The largest dataset Indochina is a web graph from
LAW\footnote{\url{http://law.di.unimi.it}}. The remaining 9 graphs, downloaded from KONECT\cite{kunegis2013konect}\footnote{\url{http://konect.uni-koblenz.de}} and SNAP~\cite{leskovec2011stanford}\footnote{\url{https://snap.stanford.edu}}, include an interaction network (WikiConflict), a coauthorship network (DBLP), a location-based social network (Gowalla), 2 web networks (Berkstan and Google), and 4 social networks (Facebook, Youtube, Petster, and Flickr). All of the graphs are unweighted. Directed graphs were converted to undirected ones in our testings. For query performance evaluation, 10,000 random queries were employed, and the average time is reported.

 
\vspace{1mm}
\noindent\textbf{Settings}~~In experiments, all programs were implemented in standard c++11 and compiled with g++4.8.5.
\\
All experiments were performed on a machine with 20X Intel Xeon 2.3GHz and 385GB main memory running Linux(Red Hat Linux 7.3 64-bit). The number of landmarks is set to $100$ by default.
\begin{table}[htpb]
  \centering
     \caption{Statistics of Datasets.}
\label{tb:datasets}
    \begin{tabular}{lrrrr}
      \hline
      \textbf{Name}	&\textbf{Dataset} 	& \textbf{$\lvert V \rvert$}	 & \textbf{$\lvert E \rvert$} 	& \textbf{$d_{avg}$} \\ \hline 
	FB		&  Facebook		&	63,731	 & 	817,035	& 	25.6			   \\ 
	GW		&  Gowalla &	196,591	 & 	950,327	& 	9.7			   \\ 
	WI		&  WikiConflict	&	118,100	 & 	2,027,871	& 	34.3			   \\ 
	GO		& Google 		&	875,713	 & 	4,322,051	& 	9.9			   \\ 
	DB		& DBLP	&	1,314,050	 & 	5,326,414	& 	8.1			   \\ 

	BE		& Berkstan 	&	685,230	 & 	6,649,470	& 	19.4			   \\ 

	YT		& Youtube 		&	3,223,589	 & 	9,375,374	& 	5.8			   \\ 
	PE		& Pester	&	623,766	 & 	15,695,166	& 	50.3			   \\ 
	FL		& Flickr	&	2,302,925	 & 	22,838,276	& 	19.8			   \\ 
	IN		& Indochina 	&	7,414,866	 & 150,984,819	& 	40.7			   \\ \hline

       \end{tabular}

\end{table}

\noindent \textbf{Exp 1: Indexing Time}.~This experiment evaluates the indexing time for three algorithms, \base, \our, and \ourp. \rrev{The ordering time is also counted in the indexing time.} What stands out in Figure~\ref{fig:runtime} is that \our \ could beat \base \ in the $7$ of $10$ datasets for single-core, except $GW$, $GO$, and $BE$. The \our \ could construct an index about 27$\%$ faster than \base \ in YT, and about $18\%$ faster on average. Moreover, \our \ could naturally be paralleled since the label entries in each round are divided into independent sets, while \base \  need to obey label dependency due to the node order rank. As for the multiple cores, the speedup of our \ourp \ could achieve nearly linear speedup with the growth of the number of cores. It is also apparent from this table that only \ourp \ could construct the index for billion-scale, which is urgently demanded in real-world applications. In the tested $10$ datasets, \ourp could achieve at least $12$ speedups when using $20$ threads compared with the single thread.  

\noindent \textbf{Exp 2: Index Size}.~Figure~\ref{fig:indexsize} shows the results on index size. What is striking about the table is that \our \  and \ourp \  return the same index size. The reason behind this phenomenon is that the dependencies are eliminated between each thread. Thus, the final index would be the same with any number of threads. As for the \base, the index size is similar to \our \ and \ourp. The reason behind this result is that the parallel paradigm does not affect index size. What is interesting about the data in this figure is that \our \ and \ourp could achieve the same index size, which indicates that there is no dependency on the index construction in each iteration. Specifically, in each iteration, the execution order of each vertex in \pull or \push does not affect the final result of our index. Figure~\ref{fig:indexsize} illustrates that \our \ and \ourp \ could construct a smaller index size in the datasets of $FB$, $GW$, $YT$, $PE$, and $FL$.

From the shortest path cover perspective, one of the most important factors for the index size is the node order of the original graph. Thus, it would compare different node orders in the following experiments.

\noindent \textbf{Exp 3: Query Time}.~Exp 3 evaluates the average query time taken by \base, \our, and \ourp \ with 100,000 random queries for each dataset. Figure~$7$ illustrates the query time of \base, \our, and \ourp. What is striking about the figures is that the query time of \base \ and \our \ is similar. Both of them could answer queries in about $100$ microseconds. Nevertheless, with the parallel techniques in \ourp, \ourp \ could achieve a nearly linear speedup when compared with \base \ and \our. \rrev{The main reason is that the query is very efficient and there is no significant bottleneck in the query process. Thus, a divide and conquer strategy on the query workload could achieve a linear speedup.}


\noindent \textbf{Exp 4: Indexing Speedup on Multi-Cores}.~The speedup of the index time of an approach on $x$ cores is calculated by using the index time of the approach with $1$ core dividing that of $x$ cores.

Thus, when the core number is $1$, the speedup is constantly $1$; when an approach fails in indexing on $1$ core within the time limit, its speedup cannot be derived. This experiment evaluates the scalability of \ourp \ by varying the number of threads on four datasets, i.e., FB, GO, GW, and WI. It is illustrated in Figure~\ref{fig:speedup_index_time} that \ourp \ could achieve nearly linear scalability with the growth of threads. When the number of threads is $20$, \ourp \ could achieve $16.7$, $11.8$, $11.9$, and $15.4$ speedups for FB, GO, GW, and WI, respectively. It is shown that the scalability of \ourp \ is better in FB and WI than that of GO and GW. According to the statistics in~\ref{tb:datasets}, FB and GW have a high average degree, i.e., 25.6 and 34.3 respectively, while GW and WI have a lower average degree, i.e., 9.7 and 9.9, respectively.
\eat{The reason behind this phenomenon is that our parallel algorithm has some extra cost, e.g., initialization, for each newly added thread. For those datasets with a low average degree, the extra cost could affect the overall performance, while in the high average-degree datasets, the extra cost could only occupy a very small percentage of the total overhead. Thus, its scalability is much better than that in the low average-degree datasets. The speedup of query time over these four datasets is also evaluated.} What stands out in Figure~\ref{fig:speedup_query_time} is that the speedup of query time could achieve nearly linear scalability with the growth of threads. A similar trend could be observed in terms of query time's scalability.

\noindent \textbf{Exp 5: Ablation Analysis}.~Exp 5 analyzes the influence of separate techniques for the shortest path counting problem in terms of scalability. It evaluates the proposed techniques, i.e., landmark labeling, schedule plan, and node order, under $20$ threads. Since these techniques do not have many effects on the indexing size and query time, it only compares their indexing time. Figure~\ref{fig:Ablation} includes three sub-figures, i.e., Figure~\ref{fig:landmark},~\ref{fig:schedule}, and~\ref{fig:nodeorder}. In Figure~\ref{fig:landmark}, \textit{LL} denotes landmark labeling, while \textit{NLL} indicates the index construction without landmark labeling. What stands out in this figure is that landmark labeling could achieve about a little faster than that without landmark labeling. Figure~\ref{fig:schedule} illustrates that our cost function-based schedule plan could achieve about somewhat faster indexing time than that of the static schedule plan. What is striking in Figure~\ref{fig:nodeorder} indicates the hybrid node order could be the fastest among these three node orders. 

\noindent\rrev{\textbf{Exp 6: The effect of $\delta$.}~Figure~\ref{fig:effect_delta} shows the effect of $\delta$ in terms of index time, index size, and query time. What stands out in this figure is that when the $\delta$ increases, the index time, index size, and query time decrease first and then increase. The reason would be that the tree decomposition order would be suitable for the vertices with small order. Thus, the $\delta$ is set to $5$ from the empirical study.}

\noindent\rrev{\textbf{Exp 7: The effect of $\#$ of landmarks.}~Figure~\ref{fig:effect_landmarks} illustrates the effect of $\#$ of landmarks. Since the landmarks do not affect the index size and query time, we only compare the index time. What is striking in this figure is that when the number of landmarks increases, the index time decreases first and then increases. The reason would be that there is an extra cost if landmark-based filtering returns a $false$ result. When the number of landmarks increases, the possibility of returning $false$ may increase.}

\noindent\rrev{\textbf{Exp 8: Break Down the Indexing Time.}~Figure~\ref{fig:breakDown} illustrates the separate time cost for the node ordering, landmark labeling (LL), and label construction (LC). What stands out in the figure is that the LC dominates all the other two phases, which is the most time-consuming part. Although the LL and Order do not cost too much time, their results have an important impact on the LC phase.}
\balance
\section{Related Works}
\label{sect:related}
In this section, some important related works are surveyed as follows.

\noindent \underline{\textit{Counting}}.~Counting the occurrences of certain structs is also fundamental in graph analytics~\cite{jain2017fast,peng10,peng2,peng11}.~\cite{pinar2017escape} present randomized algorithms with provable guarantee to count cliques and 5-vertex subgraphs in a graph, respectively. An algorithm that counts triangles in $O(m^{1.41})$ time is shown in~\cite{alon1997finding}. There are also many works on counting paths and cycles in the literature. 

The problem of exactly counting paths and cycles of length $l$, parameterized by $l$, is \#W[1]-complete under parameterized Turing reductions~\cite{flum2004parameterized}. In addition, given vertices $s$ and $t$, the problem of counting the \# of simple paths between $s$ and $t$ is \#P-complete~\cite{roberts2007estimating,valiant1979complexity}. ~\cite{bezakova2018counting} and~\cite{ren2018shortest} also study the problem of counting shortest paths for two vertices, but unlike this work, they focus on planar graphs and probabilistic networks, respectively.

\eat{\noindent \underline{\textit{Enumerating Patterns}}.~Enumerating the number of specific patterns is another method to produce the count. However, it is supposed to take much more time than direct counting \cite{giscard19}. For instance, when counting shortest cycles, one of the primary reasons is that the shortest cycle length is unknown in advance, and obtaining it requires the employment of a BFS-like method whose running time is already longer than the query time of hub labeling. Another reason is that enumeration requires finding all the vertices along with the cycles, which is unnecessary for the counting problem. With the help of a hot-point based index, \cite{qiu18} can output the newly-generated constraint cycles upon the arrival of new edges. \cite{peng19} is shown to have better time performance than \cite{qiu18} on enumerating constraint paths and cycles.}

\noindent \underline{\textit{Dynamic Maintenance for 2-hop Labeling}}.~To adapt to the dynamic update of the network, some works~\cite{akiba14,d2019fully,qin17} proposed dynamic algorithms to deal with edge insertion and deletion. For the edge insertion, a partial BFS for each affected hub is started from one of the inserted-edge endpoints and creates a label when finding the tentative distance is shorter than the query answer from the previous index. Also, hub labeling is used in the related constrained path-based problems, e.g., label constrained~\cite{peng1,peng9,peng13}.
\balance
\section{Future Work and Conclusion}
\label{sect:con}

\noindent\rrev{\textbf{Future Work.}~Experiments demonstrate that for graphs such as DB and FL, our techniques require much more indexing time and index space than for other graphs of comparable size. Unfortunately, it is currently unclear what properties of these graphs contribute to this inefficiency. It would be interesting to investigate it in our future work.}

\noindent \textbf{Conclusion}.~\eat{In graph analytics, the concept of shortest path is fundamental. While numerous works have been devoted to developing efficient distance oracles for computing the shortest path distance between any vertices and $t$, we examine the problem of efficiently counting the number of shortest paths between $s$ and $t$ in light of its applications in tasks such as betweenness analysis. Counting the shortest paths is a key issue in the graph database field. Numerous studies have been conducted in recent years to attempt to resolve such problems. Nonetheless, the existing approach confronts a significant obstacle in the form of large graphs, limiting its application possibilities.}
We study the problem of counting the $\#$ of shortest paths between two vertices $s$ and $t$ in the context of parallel. To address the scalability issue of existing work, we propose a parallel algorithm to speedup this problem. We also investigate the proper vertex ordering for the parallel index construction. Our comprehensive experimental study verifies the effectiveness and efficiency of our algorithms.  

\section*{Acknowledgment}
This work is supported by Hong Kong RGC GRF grant (No. 14203618, No. 14202919,  No. 14205520, and No. 14217322), RGC CRF grant (No. C4158-20G), Hong Kong ITC ITF grant (No. MRP/071/20X), and NSFC grant (No. U1936205). 

\bibliographystyle{ieeetr}
{
\small
\bibliography{ref}
}

\end{document}